\documentclass[10pt,twocolumn]{IEEEtran}

\usepackage{amsmath,graphics,amssymb,epsfig,subfigure,color,cite}

\usepackage{array}
\usepackage{multirow}
\usepackage{pifont}
\usepackage{arydshln}
\usepackage{xcolor}
\usepackage{enumerate}
\usepackage{bm}
\usepackage{float}
\usepackage{makecell}
\usepackage[right=0.625in, left=0.625in, top=0.7in, bottom=1.1in]{geometry}

\usepackage{lipsum}
\usepackage{capt-of}

\usepackage{algorithm}
\usepackage{algpseudocode}
\usepackage{algcompatible}
\algblockdefx{WHILE}{ENDWHILE}[1]%
  {\textbf{while }#1 \textbf{}}%
  {\textbf{end}}
\algblockdefx{FORP}{ENDFORP}[1]%
  {\textbf{for }#1 \textbf{do in parallel}}%
  {\textbf{end}}
\algblockdefx{FOR}{ENDFOR}[1]%
  {\textbf{for }#1 \textbf{do}}%
  {\textbf{end}}
\algblockdefx{IF}{ENDIF}[1]%
  {\textbf{if }#1 \textbf{}}%
  {\textbf{end}} 
\algblockdefx{ELSIF}{ENDIF}[1]%
  {\textbf{else if }#1 }%
  {\textbf{end}}
\algnewcommand\algorithmicprocedure{\textbf{function}}
\algnewcommand\FUNC{\item[\algorithmicprocedure]}%
\algnewcommand\algorithmicendprocedure{\textbf{end function}}
\algnewcommand\ENDFUNC{\item[\algorithmicendprocedure]}%
 \usepackage{setspace}
\let\Algorithm\algorithm
\renewcommand\algorithm[1][]{\Algorithm[#1]\setstretch{1.4}}

\usepackage{bbm}
\usepackage{bigints}
\usepackage{adjustbox}

\newcommand{\cmark}{\ding{51}}
\newcommand{\xmark}{\ding{55}}

\floatname{algorithm}{Algorithm}

\usepackage{enumitem}

\newcommand{\argmin}{\operatornamewithlimits{argmin}}
\makeatletter
\newcommand{\vast}{\bBigg@{4.5}}
\newcommand{\Vast}{\bBigg@{7.5}}
\makeatother

\begin{document}
\title{\fontsize{23}{28}\selectfont Semantic Feature Channel Optimization: A Unified Framework for Analog and Digital Semantic Communications
}

\author{Yongjeong Oh, \IEEEmembership{Member,~IEEE}, Jihong Park, \IEEEmembership{Senior Member,~IEEE}, \\Jinho Choi, \IEEEmembership{Fellow,~IEEE}, and Yo-Seb Jeon, \IEEEmembership{Member,~IEEE}\vspace{-3mm}
\thanks{\emph{(Corresponding author: Yo-Seb Jeon)}.}	    
\thanks{Yongjeong Oh and Jihong Park are with the Information Systems Technology and Design Pillar, Singapore University of Technology and Design, Singapore 487372 (email: \{yongjeong\_oh, jihong\_park\}@sutd.edu.sg).}
        \thanks{Jinho Choi is with the School of Electrical and Mechanical Engineering, Adelaide University, SA 5005, Australia (email: jinho.choi@adelaide.edu.au).}
	\thanks{Yo-Seb Jeon is with the Department of Electrical Engineering, POSTECH, Pohang, Gyeongbuk 37673, Republic of Korea (email:  yoseb.jeon@postech.ac.kr).}
        }
	\vspace{-2mm} 
	
	\maketitle
	\vspace{-12mm}
\begin{abstract} 
Semantic communication (SC) aims to improve communication efficiency by transmitting only task-relevant information, termed semantic features (SFs). Most existing SC frameworks achieve this by optimizing the encoder and decoder while treating the channel between them as fixed. In this paper, we introduce a new perspective by explicitly modeling the entire process from the encoder output to the decoder input as an SF channel, resulting in an encoder--SF channel--decoder pipeline. We observe that the SF channel is configurable through transceiver operations, such as power allocation, thereby providing an additional degree of freedom for improving task performance. Inspired by this, we formulate a joint optimization problem for the encoder, SF channel, and decoder under a mutual information constraint between the transmitted and reconstructed SFs. To provide analytical insight, we derive the optimal SF channel in closed form for an analytically tractable setting. 
Based on this formulation, we develop a unified semantic feature channel optimization framework applicable to both analog and digital SC systems. To realize the SF channel in general communication systems, we further propose a physical-layer calibration strategy that aligns the actual SF channel with the trained one. Simulation results demonstrate that the proposed framework consistently improves task performance across various communication environments. 
\end{abstract}


\section{Introduction}\label{Sec:Intro}
With recent advances in artificial intelligence (AI), next-generation wireless networks are expected to support emerging AI applications such as digital twins, intelligent transportation, and collaborative robotics \cite{10845817,SemCom_Survey_1}. These applications require frequent and large-scale data exchange, imposing a heavy burden on existing communication systems that are designed to accurately transmit raw data. Fortunately, as AI systems become increasingly capable of perception and decision-making, the communication goals in such applications are gradually changing from accurately reproducing raw data to conveying information that contributes to performing specific tasks. This shift has led to the emergence of a new communication paradigm known as semantic communication (SC), which focuses on transmitting task-relevant information, thereby improving bandwidth efficiency and robustness against channel perturbations \cite{SemCom_Survey_2,SemCom_Survey_3,Semantic_Jihong_1}.

In the literature, SC is implemented through various approaches, including deep separate source-channel coding (DeepSSCC) \cite{10175391}, generative SC \cite{10628028}, and deep joint source-channel coding (DeepJSCC) \cite{DeepJSCC}. Among these, DeepJSCC has attracted significant attention due to its simple yet high task performance, where the source and channel encoders/decoders are integrated into a unified neural-type encoder/decoder. 
Specifically, in this approach, the encoder at the transmitter transforms the input data into a latent representation, referred to as the \textit{semantic features (SFs)}. These SFs are then transmitted through either analog or digital communication, depending on whether the SFs are conveyed in continuous or discrete forms. In both cases, the transmitted SFs inevitably experience degradation due to channel fading and noise. At the receiver, the decoder performs the designated task based on the received SFs. Here, we note that the entire process from the encoder output to the decoder input, including the physical channel, modulation, power control, and other transceiver operations, can be collectively modeled as an equivalent channel, referred to as the \textit{SF channel}. 
As illustrated in Fig.~\ref{fig:SFC}, This abstraction allows us to interpret the DeepJSCC framework as an \textit{encoder--SF channel--decoder} pipeline, where the SF channel represents how the transmitted SFs are distorted and transformed through the communication process.

\begin{figure*}[t]
    \centering
    \subfigure[Analog SC pipeline with the SF channel $p(\hat{\bm z}|{\bm z})$]{
    {\epsfig{file=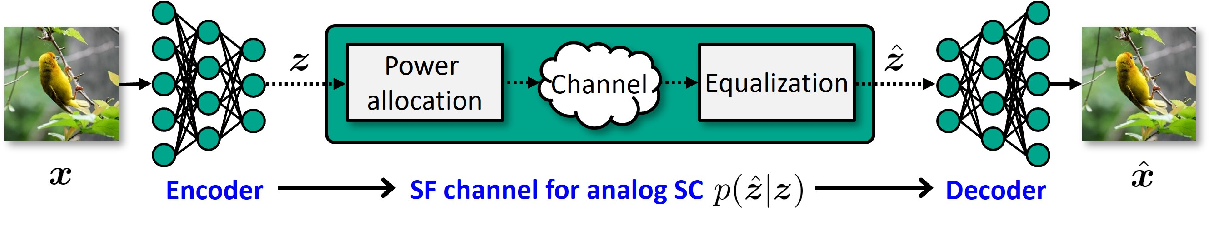, width=10cm}}
    }\vspace{0.05cm}
    \subfigure[Digital SC pipeline with the SF channel $p(\hat{\bm b}|{\bm b})$]{
    {\epsfig{file=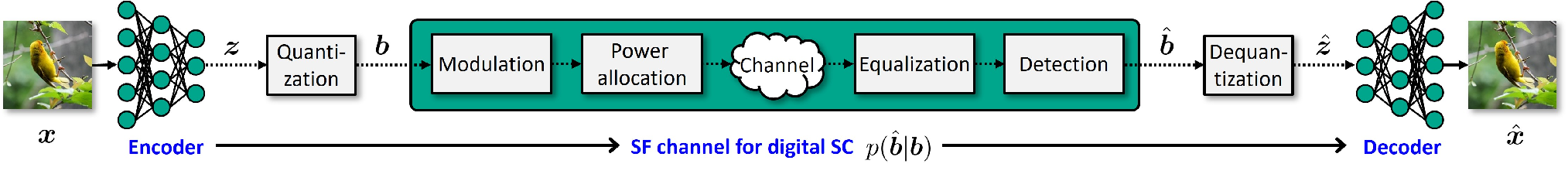, width=18cm}}
    }\\
    {\caption{The overall analog and digital SC pipelines: Encoder--SF channel--decoder.}\label{fig:SFC}}\vspace{-3mm}
\end{figure*}

Over the years, extensive research has focused on encoder-decoder (enc-dec) optimization under a single fixed SF channel assumption \cite{DeepJSCC,DeepJSCC_Q,OFDM_SC,10278812,NECST}.
For example, earlier works modeled the SF channel as an additive white Gaussian noise (AWGN) or Rayleigh fading channel with a fixed signal-to-noise ratio (SNR) level \cite{DeepJSCC,DeepJSCC_Q}. 
This approach was further extended by incorporating multipath fading effects and multi-antenna systems into the SF channel \cite{OFDM_SC,10278812}.
Also, in \cite{NECST}, to capture bit-level transmission characteristics of digital communication, the SF channel was characterized as a binary symmetric channel (BSC) with a fixed bit-flip probability. These studies have successfully demonstrated the effectiveness of enc-dec-centric optimization by achieving high task performance under the assumed SF channel. Nevertheless, the resulting enc-dec often suffers from significant performance degradation when the actual SF channel deviates from the trained one. 

To address this challenge, several studies have attempted to enhance robustness or adaptability by training the enc-dec under multiple SF channels \cite{DeepJSCC_Analog_1,Joohyuk,JSCC_universal,10436878,10382545}.
In \cite{DeepJSCC_Analog_1}, multiple SF channels were configured under an AWGN or Rayleigh fading channel with varying SNR levels. 
For digital communications, multiple SF channels were generated by sampling the bit-flip probabilities of the BSCs \cite{Joohyuk} or by changing both the SNR levels and modulation orders under an AWGN channel \cite{JSCC_universal}. Further, in \cite{10436878} and \cite{10382545}, multiple SF channels were induced by varying the number of transmitted SFs for rate-adaptive SCs.
Although these studies further improve the robustness and generality of the enc-dec-centric optimization, their performance still tends to degrade under unseen communication environments. 
Moreover,  covering all possible SF channels using the above approaches would require excessive data sampling and/or model complexity as the actual SF channel can vary by numerous factors, including antenna configurations, channel statistics, and interference/noise levels.


Unlike prior works that focus solely on enc–dec-centric optimization under fixed SF channels, this work highlights that the SF channel is configurable, as its behavior depends on various transceiver operations such as power allocation and adaptive modulation and coding \cite{goldsmith2005wireless}. This implies that the SF channel itself can be incorporated into the training process to further improve task performance. 
An initial attempt in this direction was made in \cite{BlindSC}, which jointly optimizes the enc–dec and a BSC-based SF channel. However, the applicability of \cite{BlindSC} is limited to point-to-point digital SC scenarios, and relies on a heuristic regularization loss for optimizing the SF channel. 
Table~\ref{table:related_works} summarizes representative studies on SF channel modeling and optimization and categorizes them into three groups: single SF channels, multiple SF channels, and trainable SF channels.

To further shed light on the potential of configuring the SF channel, this paper proposes a general joint optimization framework applicable to both analog and digital SCs. 
Specifically, we first formulate the SF channel optimization problem to maximize task performance under a mutual information constraint between the transmitted and reconstructed SFs. We then analyze a tractable case with a Gaussian source and linear enc-dec mappings to provide analytical insight. Based on this formulation, we propose an end-to-end training strategy that jointly optimizes the non-linear deep neural network (DNN)-based enc-dec and the SF channel. We also introduce a communication strategy that realizes the trained SF channel in practical communication scenarios by controlling physical layer (PHY) parameters including transmit power and modulation level. Overall, by optimizing the SF channel, the proposed framework controls the distortion of each SF, thereby improving the task performance of SC. 

The major contributions of this paper are summarized as follows.

\begin{table*}[t]
\centering
\caption{Comparison of related works on SF channel modeling and optimization.}
\label{table:related_works}
\small
\resizebox{\textwidth}{!}{%
\begin{tabular}{l l c c c c c}
\hline
Category & Work & SF channel & \makecell{Channel adaptiveness \\ or robustness} & \makecell{Train SF channel \\ for analog SCs} & \makecell{Train SF channel \\ for digital SCs} & \makecell{Analytical evidence \\ for SF channel} \\
\hline
\hline
\multirow{7}{*}{\makecell{Single SF channel}}
& \cite{DeepJSCC,DeepJSCC_Q} & \makecell{AWGN/Rayleigh \\ with a fixed SNR} & \xmark & \xmark & \xmark & \xmark \\
\cdashline{2-7}
& \cite{OFDM_SC}  & \makecell{OFDM with a fixed \\ number of multipaths} & \xmark & \xmark & \xmark & \xmark \\
\cdashline{2-7}
& \cite{10278812}  & \makecell{MIMO with a fixed \\ number of antennas} & \xmark & \xmark & \xmark & \xmark \\
\cdashline{2-7}
& \cite{NECST}  & \makecell{BSC with a fixed \\ bit-flip probability} & \xmark & \xmark & \xmark & \xmark \\
\hline
\hline
\multirow{7}{*}{\makecell{Multiple SF channels}}
& \cite{DeepJSCC_Analog_1}  & \makecell{AWGN/Rayleigh \\ with varying SNRs} & \cmark & \xmark & \xmark & \xmark \\
\cdashline{2-7}
& \cite{Joohyuk}  & \makecell{BSC with varying \\ bit-flip probabilities} & \cmark & \xmark & \xmark & \xmark \\
\cdashline{2-7}
& \cite{JSCC_universal}  & \makecell{AWGN with varying \\ SNRs and modulation orders} & \cmark & \xmark & \xmark & \xmark \\
\cdashline{2-7}
& \cite{10436878,10382545}  & \makecell{Transmission with \\ varying
numbers of SFs} & \cmark & \xmark & \xmark & \xmark \\
\cdashline{2-7}
\hline
\hline
\multirow{3}{*}{\makecell{Trainable SF channel}}
& \cite{BlindSC}  & \makecell{Trainable SF channel for digital SCs \\ with a heuristic loss function} & \cmark & \xmark & \cmark & \xmark \\
\cdashline{2-7}
& \textbf{This work} & \makecell{Trainable SF channels for analog and digital SCs \\ with a mutual information constraint} & \cmark & \cmark & \cmark & \cmark \\
\hline
\end{tabular}
}
\vspace{-3mm}
\end{table*}

\begin{itemize}
    \item We formulate the joint optimization problem of the enc-dec and the SF channel to maximize task performance in both analog and digital SCs. In this problem, we introduce a mutual information constraint between the transmitted and reconstructed SFs to prevent convergence to a trivial error-free SF channel, while accounting for various communication constraints in an integrated manner. To provide analytical insight, we also derive the optimal SF channel in closed form for an analytically tractable case with a linear enc-dec and a Gaussian source.

    \item We propose an end-to-end training strategy for jointly optimizing the DNN-based enc-dec and the SF channel under limited mutual information. In analog SCs, we introduce a differentiable surrogate model based on an AWGN, where each SF is corrupted by Gaussian noise with a learnable variance. In digital SCs, we present a differentiable surrogate model based on relaxed BSCs, where each bit can be flipped with a learnable bit-flip probability. In both analog and digital SCs, the limited mutual information is addressed as a rate allocation problem, in which each SF or bit is assigned a portion of the total communication rate. The allocated rate determines the noise variance in analog SCs or the bit-flip probability in digital SCs, thereby enabling individual control over the distortion of each SF or bit. 
    

    \item We introduce a communication strategy, referred to as a PHY calibration strategy, which realizes the optimized SF channel by controlling PHY parameters. In single-user analog SCs, the proposed strategy determines the transmit power and feature-to-channel mapping so that the actual SNR matches the trained SNR. 
    This approach can be readily extended to multi-user analog SCs. In multi-user digital SCs, it jointly adjusts the transmit powers and modulation levels across multiple users to align the actual BERs with the trained bit-flip probabilities. For both analog and digital SCs, the proposed strategy selects the most suitable SF channel among multiple candidates to adapt to varying communication environments. 
    
    \item Through simulation, we demonstrate that the proposed framework achieves superior image reconstruction quality across various mutual information limits. 
    Furthermore, we show that the proposed PHY calibration strategy realizes the optimized SF channel in actual wireless environments. 
\end{itemize}

This work extends our earlier conference paper in \cite{Conf_ver}, which focused on Gaussian-source analysis and analog SC. The present work further develops a training framework for digital SC, a PHY calibration strategy for multi-user digital SC, a detailed comparison with \cite{BlindSC}, and extensive evaluations across multiple datasets and baselines. 


\section{System Model and Concept of SF Channel}\label{Sec:System}
In this section, we first present the analog and digital SC systems and then introduce the concept of the SF channel.

\subsection{System Model}\label{Sec:Sys}
We consider a typical SC model where a transmitter is connected to a receiver over a wireless network to perform an image reconstruction task.
This model can be readily extended to other machine learning tasks and SC architectures.
Let ${\bm x} \in \mathbb{R}^N$ denote the image data of length $N$.
The transmitter encodes ${\bm x}$ using an encoder as follows:
\begin{align}
{\bm z}=f_{{\bm \theta}_{\rm enc}}({\bm x})\in \mathbb{R}^M,
\end{align}
where $f_{{\bm \theta}_{\rm enc}}(\cdot)$ is the encoding function parameterized by ${\bm \theta}_{\rm enc}$, and ${\bm z}$ is the SF vector of length $M$. 

After encoding, the SF vector ${\bm z}$ is mapped into either an analog or a digital symbol depending on whether analog or digital communication is employed. 
\begin{itemize}
    \item Analog symbol mapping: Each SF is mean-centered and scaled as 
    \begin{align}
        \tilde{z}_m^{\rm (A)} = \sqrt{p_m^{\rm (A)}}({z}_m-\mu_m),\quad m\in\{1,\cdots,M\},
    \end{align}
    where $\mu_m$ is the mean of the $m$-th SF, and $p_m^{\rm (A)}$ is a power allocation coefficient satisfying $\sum_{m} \mathbb{E}[|\tilde{z}_m^{\rm (A)}|^2] \leq P_{\rm tot}$, with $P_{\rm tot}$ denoting the total power budget.  
    Then, pairs of real-valued SFs are grouped into complex symbols as
    \begin{align}
        {s}_u^{\rm (A)} = \tilde{z}_{2u-1}^{\rm (A)} + j\,\tilde{z}_{2u}^{\rm (A)},\quad u\in\left\{1,\cdots,\frac{M}{2}\right\}.
    \end{align}
    
    \item Digital symbol mapping: The SF vector is first quantized into a bit sequence ${\bm b} \in \{0,1\}^B$ of length $B$ using standard quantization methods \cite{oh2023communication,Linear_quant}. The bit sequence is then mapped to a symbol sequence $\tilde{\bm z}^{\rm (D)} \in \mathcal{C}^T$ of length $T$ through a digital modulation process, where $\mathcal{C}$ denotes the constellation set. Each modulated symbol is scaled as
    \begin{align}
        {s}_t^{\rm (D)} = \sqrt{p_t^{\rm (D)}}\tilde{z}_t^{\rm (D)},\quad t\in\left\{1,\cdots,T\right\},
    \end{align}
    where $p_t^{\rm (D)}$ is the transmit power allocated to the $t$-th symbol $\tilde{z}_t^{\rm (D)}$ and satisfies $\sum_{t} p_t^{\rm (D)} \leq P_{\rm tot}$ under the assumption that $\mathbb{E}[|\tilde{z}_t^{\rm (D)}|^2] = 1$.
\end{itemize}
For consistency, we denote the symbol sequence length by $T$. In analog mapping, $T = M/2$ is fixed, while in digital mapping, $T$ varies depending on the modulation order. The superscripts $(\rm A)$ and $(\rm D)$ are omitted hereafter for notational simplicity.

Under a flat-fading channel, the received signal at the $t$-th channel use is expressed as 
\begin{align}\label{eq:analog_y}
    y_t = h_t{s}_t + n_t,
\end{align}
where $h_t \in \mathbb{C}$ is the channel coefficient, and $n_t\sim\mathcal{CN}(0,\sigma^2)$ is AWGN with variance $\sigma^2$. The channel coefficient $h_t$ may remain constant or vary depending on the coherence time and the number of subcarriers \cite{goldsmith2005wireless}. Upon receiving the signal in \eqref{eq:analog_y}, the receiver performs channel equalization to obtain the equalized signal at the $t$-th channel use, expressed as 
\begin{align}\label{eq:equal_y}
    \tilde{y}_t \triangleq \frac{h_t^*}{|h_t|^2}y_t={s}_t + \tilde{n}_t,
\end{align}
where $\tilde{n}_t \sim \mathcal{CN}(0,\frac{\sigma^2}{|h_t|^2})$. From the equalized signal in \eqref{eq:equal_y}, an estimate of $z_m$ is obtained using either an analog or digital demapping process. 
\begin{itemize} 
    \item Analog symbol demapping: The equalized signal $\tilde{y}_t$ is decomposed into its in-phase and quadrature components, followed by power de-scaling and mean restoration. This standard receiver-side processing leads to the equivalent additive-noise model, given by 
    \begin{align}\label{eq:analog_eff_received}
        \hat{z}_m = z_m + w_m,\quad w_m \sim \mathcal{N}\!\left(0,\tfrac{\sigma^2}{2|h_t|^2p_m}\right),
    \end{align}
    where $t=\lceil m/2 \rceil$. 

    \item Digital symbol demapping: Symbol detection is performed on $\tilde{y}_t$ to recover the estimated bit sequence $\hat{\bm b}\in\{0,1\}^B$. The estimated bit sequence is then dequantized to obtain the estimated SF vector $\hat{\bm z}$. 
\end{itemize}
Finally, the receiver reconstructs an image using a decoder as follows:
\begin{align}
\hat{\bm x}=f_{{\bm \theta}_{\rm dec}}(\hat{\bm z})\in \mathbb{R}^N,
\end{align}
where $\hat{\bm x}$ denotes the reconstructed image, and $f_{{\bm \theta}_{\rm dec}}(\cdot)$ represents the decoding function parameterized by ${\bm \theta}_{\rm dec}$.

\subsection{SF Channel}
{\bf Definition (SF channel):} The \textit{SF channel} is the equivalent channel between the encoder output and the decoder input, denoted by $p(\hat{\bm z}|{\bm z})$ for analog communication and by $p(\hat{\bm b}|{\bm b})$ for digital communication\footnote{In this work, we focus on digital SC systems, where the encoder output is quantized and converted into a bit sequence. Nevertheless, the system can be extended to the discrete symbol domain, as in \cite{DeepJSCC_Q}, where the SF channel can be modeled as $p(\hat{\bm s}|{\bm s})$ with $\bm s$ denoting discrete modulation symbols.}. 
\vspace{1mm}

For the system described in Sec.~\ref{Sec:Sys}, the SF channel includes the entire transmit-receive process, including power control and equalization. Following the definition of the SF channel, the overall SC pipeline can be represented as 
\begin{align}
    {\bm x}
    \xrightarrow{\text{Enc}}
    {\bm z}~(\text{or}~{\bm b})
    \xrightarrow{\text{SF channel}}
    \hat{\bm z}~(\text{or}~\hat{\bm b})
    \xrightarrow{\text{Dec}}
    \hat{\bm x},
\end{align}
where its visualization is shown in Fig.~\ref{fig:SFC}. 

\section{Motivation and Case Study}\label{Sec:Comp}
In this section, we present the motivation for joint enc-dec and SF channel optimization and provide a case study that analytically illustrates its necessity. 

\subsection{Motivation for Joint Optimization}\label{Sec:Motivation}
Our key observation is that the SF channel is {\em configurable} through various communication strategies. For instance, \eqref{eq:analog_eff_received} shows that the distortion of each SF $z_m$ can be controlled by adjusting the corresponding power coefficient $p_m$. 
In digital communication systems, the modulation level provides another means of control, enabling a trade-off between the number of channel uses $T$ and total transmit power $P_{\rm tot}$. 
Despite this inherent configurability, most existing works focus solely on optimizing the enc-dec while keeping the SF channel fixed. This motivates us to consider the joint optimization of the enc-dec and the SF channel.

A trivial SF channel that maximizes task performance corresponds to the ideal {\em error-free} channel, i.e., $\hat{\bm z} = {\bm z}$. While this solution preserves all SFs without distortion, it fails to account for the task-dependent importance of individual SF components and results in an inefficient allocation of communication resources. In particular, over-allocating resources to all SFs disregards their heterogeneous contributions to task performance, leading to potentially excessive resource usage. Therefore, it is essential to impose appropriate constraints on the SF channel, ensuring operation under limited communication resources and within a {\em non-trivial} regime where $\hat{\bm z} \neq {\bm z}$.

To facilitate optimal SF channel design under resource constraints, we raise the following motivating question:


\vspace{1mm}
{\bf Motivating Question:} What is the optimal SF channel that maximizes task performance under a limited mutual information, i.e., $I({\bm z};\hat{\bm z}) \leq I_{\rm max}$ or $I({\bm b};\hat{\bm b}) \leq I_{\rm max}$?
\vspace{1mm}

From an information-theoretic perspective, this constraint implies that the transmitted and received SFs are not identical. Therefore, it indirectly imposes communication constraints and allows us to accommodate various communication scenarios. One representative example is the Gaussian channel, where the mutual information is determined by bandwidth, transmit power, and noise variance. By constraining the mutual information, these factors are restricted in an integrated manner. As a result, the framework captures the effects of multiple communication constraints, while avoiding excessive training overhead for the enc-dec.
The mutual information limit $I_{\rm max}$ does not represent the Shannon channel capacity; rather, it quantifies how much information about ${\bm z}$ can be conveyed through the channel to $\hat{\bm z}$. A smaller $I_{\rm max}$ indicates more severe degradation under poor channel conditions or low transmit power, whereas a larger $I_{\rm max}$ corresponds to more reliable transmission. Since channel-induced distortions are inherently allowed, our formulation differs from the classical source-channel separation theory.

\begin{figure*}[t]
    \centering 
    \subfigure[Source and noise variances]{
    {\epsfig{file=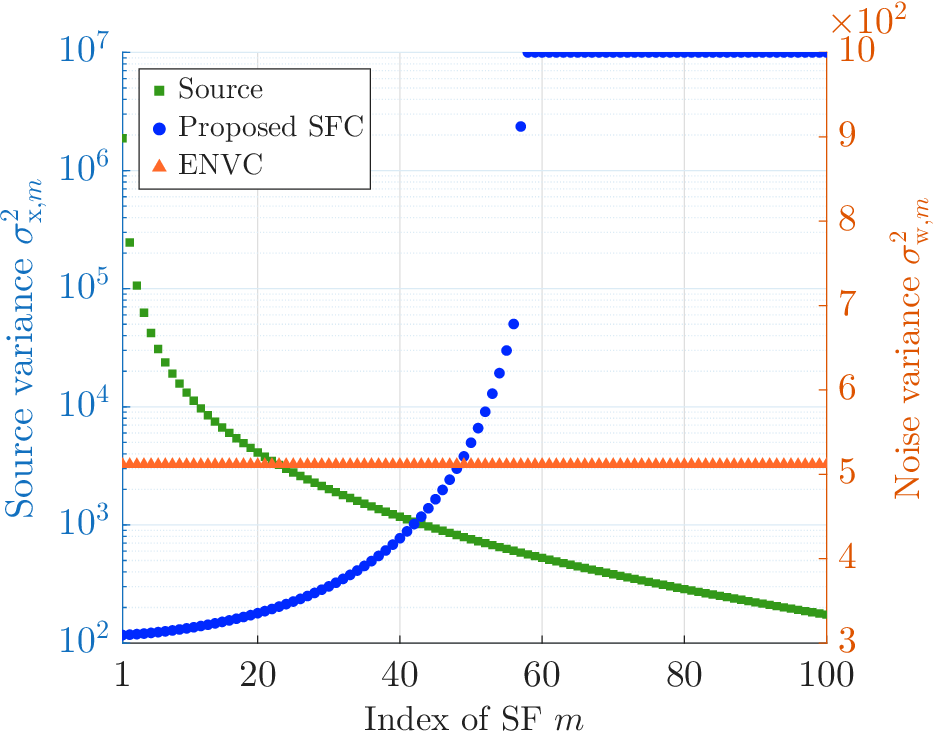, width=5.8cm}}
    }\hspace{0.01cm}
    \subfigure[MSE vs. $I_{\rm max}$]{
    {\epsfig{file=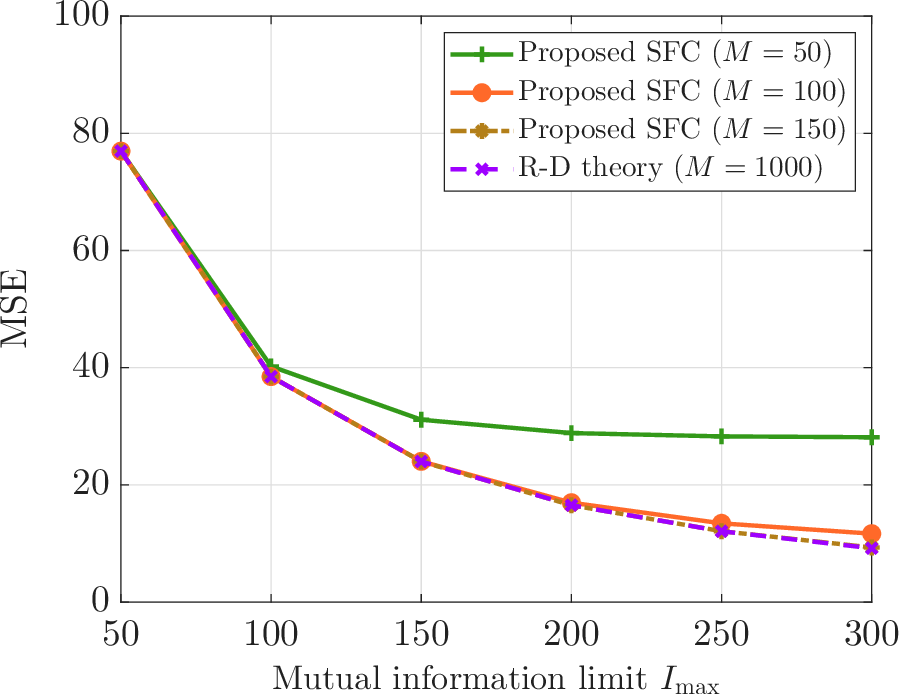, width=5.8cm}}
    }\hspace{0.01cm}
    \subfigure[MSE vs. $M$]{
    {\epsfig{file=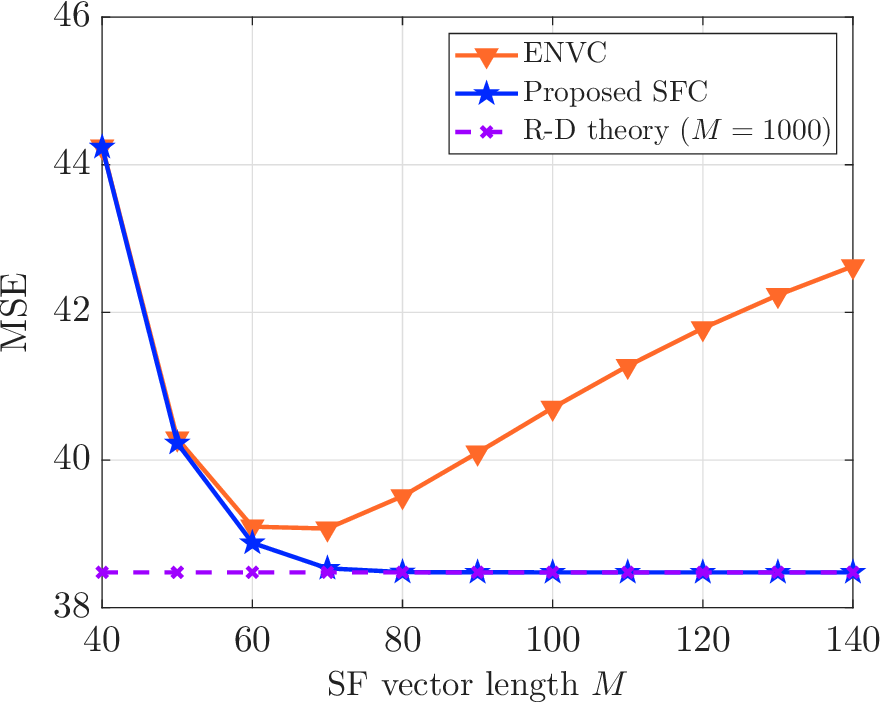, width=5.8cm}}
    }\\
    {\caption{Numerical results for the case study: source/noise variances across SFs, and the resulting MSE performance over $I_{\rm max}$ and $M$.}\label{fig:Gaussian_C_M}}\vspace{-3mm}
\end{figure*}

\subsection{Case Study: Analog SC with Linear Enc-Dec and Gaussian Input}\label{Sec:Gauss_Linear}
Before addressing the motivating question in general SC scenarios involving nonlinear DNN-based enc-dec and unknown input distributions, it is important to establish that this joint optimization is a well-defined problem that yields a meaningful solution. Without this verification, the framework proposed in Sec. IV would lack an analytical evidence for jointly optimizing the enc-dec and the SF channel, rather than optimizing the enc-dec alone under a fixed SF channel. To establish this, we consider an analytically tractable scenario with a linear enc-dec and a Gaussian source, for which the joint optimization can be characterized in closed form.



Let ${\bm x}\in\mathbb{R}^N$ be a Gaussian source with distribution ${\bm x}\sim\mathcal{N}({\bm 0},{\bm \Sigma}_{\rm xx})$, where ${\bm \Sigma}_{\rm xx}={\rm diag}(\sigma_{{\rm x},1}^2,\cdots,\sigma_{{\rm x},N}^2)$ and $\sigma_{{\rm x},1}^2 \geq \cdots \geq \sigma_{{\rm x},N}^2$. 
An encoder compresses ${\bm x}$ as ${\bm z}={\bm A}{\bm x}\in\mathbb{R}^M$, with $M\leq N$ and ${\bm A}\in\mathbb{R}^{M\times N}$ satisfying ${\bm A}{\bm A}^{\sf T}={\bm I}_M$ to constrain the encoder output power.
The SF channel is assumed to add an independent Gaussian noise ${\bm w}\sim\mathcal{N}({\bm 0},{\bm \Sigma}_{\rm ww})$, where ${\bm \Sigma}_{\rm ww}={\rm diag}(\sigma_{{\rm w},1}^2,\cdots,\sigma_{{\rm w},M}^2)$, resulting in the received signal $\hat{\bm z} = {\bm z} + {\bm w}$. 
A decoder reconstructs $\hat{\bm x}={\bm B}\hat{\bm z}$, where ${\bm B}\in\mathbb{R}^{N\times M}$. The optimization problem is formulated as 
\begin{align}
    ({\bf P0})~~\min_{ {\bm A}, {\bm B}, {\bm \Sigma}_{\rm ww} }& \mathbb{E}[\| {\bm x} - \hat{\bm x} \|^2], \label{eq:P3_obj} \\
    \text{s.t.}~~~&I({\bf z};\hat{\bf z}) \leq I_{\rm max},~{\bm A}{\bm A}^{\sf T} = {\bm I}_M.\label{eq:P0_const}
\end{align}
Consequently, the optimal solution of Problem ${\bf P0}$ is given by 
\begin{align}
    {\bm A}^\star &= \big[{\bm I}_M,~{\bm 0}_{M\times (N-M)}\big], \label{eq:opt_A} \\
    {\bm B}^\star &= {\bm \Sigma}_{\rm xx}{\bm A}^{{\star}{\sf T}}
    \big({\bm A}^\star {\bm \Sigma}_{\rm xx}{\bm A}^{{\star}{\sf T}} + {\bm \Sigma}_{\rm ww}^\star \big)^{-1}, \label{eq:opt_B} \\
    {\bm \Sigma}_{\rm ww}^\star &= {\rm diag}\left(\big(\sigma_{{\rm w},1}^2\big)^\star,\cdots,\big(\sigma_{{\rm w},M}^2\big)^\star\right), \label{eq:opt_Sigma}
\end{align}
where 
\begin{align}\label{eq:opt_n_var}
    \big(\sigma_{{\rm w},m}^2\big)^\star =
    \begin{cases}
    \displaystyle \frac{\lambda^\star \sigma_{{\rm x},m}^2}{2\sigma_{{\rm x},m}^2 - \lambda^\star} & \text{if } \lambda^\star < 2\sigma_{{\rm x},m}^2, \\
    \infty & \text{if } \lambda^\star \geq 2\sigma_{{\rm x},m}^2,
    \end{cases}
\end{align}
and $\lambda^\star$ is the optimal Lagrangian multiplier chosen so that the mutual information constraint in \eqref{eq:P0_const} holds with equality. 
The derivation of this solution can be summarized as follows: (i) substituting the optimal LMMSE decoder of \eqref{eq:opt_B} into the objective and reducing the objective to a function of only ${\bm A}$ and ${\bm \Sigma}_{\rm ww}$; (ii) applying a majorization and Schur-convexity argument to establish the optimal form of ${\bm A}$, where each column is either a zero vector or a one-hot vector; (iii) applying the Lagrangian method to obtain the optimal form of ${\bm \Sigma}_{\rm ww}$; and (iv) substituting this form back into the objective and comparing the resulting values across different combinations of zero and one-hot vectors of ${\bm A}$ to show that the objective is minimized by \eqref{eq:opt_A} and \eqref{eq:opt_Sigma}. A detailed derivation of these steps can be found in the conference version of this work~\cite{Conf_ver}.

Fig.~\ref{fig:Gaussian_C_M} shows the numerical results for the above solution. In the simulation, we set $N=1000$ and $\sigma_{{\rm x},n}^2 \sim {\rm Lognormal}(0,4)$, with the variances sorted in descending order. We refer to the above solution as the proposed SFC, and compare it against two baseline schemes: ENVC (an equal-noise-variance channel across all SFs with the optimal enc-dec) and R-D theory (Gaussian R-D bound with ${\bm A}={\bm B}={\bm I}_N$ \cite{cover1999elements}). Fig.~\ref{fig:Gaussian_C_M}(a) shows the source and noise variances when $M=100$ and $I_{\rm max}=100$, where the noise variance is clipped to $10^3$ for visualization. The result demonstrates that the proposed SFC allocates a low noise variance to high-variance sources and a high noise variance to low-variance sources. This behavior leads to the superior MSE performance, as shown in Figs.~\ref{fig:Gaussian_C_M}(b) and (c). 
Specifically, Fig.~\ref{fig:Gaussian_C_M}(b) illustrates the MSE performance over $I_{\rm max}$. The result shows that the proposed SFC closely follows the R-D bound, with a negligible gap for moderate $M$. Fig.~\ref{fig:Gaussian_C_M}(c) presents the MSE performance over $M$ when $I_{\rm max}=100$. The result shows that the proposed SFC rapidly converges to the R-D bound as $M$ increases, while ENVC degrades for large $M$. This degradation occurs because, under ENVC, the noise variance increases with $M$, and this increased noise is applied uniformly across sources without accounting for their variance, thereby causing greater distortion to high-variance sources as well. 

\vspace{1mm}
{\bf Remark 1 (Connection to R-D Theory):} The R-D bound can be derived using the classical test channel in \cite{cover1999elements}, where the source is directly regarded as the channel input, i.e., $M=N$ and ${\bm A}={\bm B}={\bm I}_N$. In contrast, Problem ${\bf P0}$ considers a more general setting with $M\leq N$. The optimal solution in \eqref{eq:opt_A}--\eqref{eq:opt_Sigma} reduces to the classical test-channel result when $M=N$, which is consistent with the numerical results in Fig.~\ref{fig:Gaussian_C_M}, where the performance approaches the R-D bound as $M$ approaches $N$. However, deriving \eqref{eq:opt_A}--\eqref{eq:opt_Sigma} directly from the classical test-channel result is not straightforward, since the latter presupposes ${\bm A}={\bm B}={\bm I}_N$. Here, we emphasize that the value of this result does not lie in extending classical rate-distortion theory, but in showing that the joint optimization of the enc-dec and the SF channel yields a meaningful and non-trivial solution even in this simplified setting.

\section{Proposed Joint Enc-Dec and SF Channel Optimization for SCs}\label{Sec:Train}
In Sec.~\ref{Sec:Comp}, we showed that the joint optimization of the enc–dec and the SF channel constitutes a well-defined optimization problem with a non-trivial solution. Although this analysis provides theoretical justification for the proposed formulation, deriving the optimal SF channel analytically is generally intractable in practical SC systems due to unknown input distributions and nonlinear DNN-based enc–dec models. To overcome this challenge, this section presents an end-to-end training framework by introducing differentiable surrogate models for the analog and digital SF channels. The high-level procedure of our strategy is illustrated in Fig.~\ref{fig:Train_SF}.

\subsection{End-to-End Training for Analog SC}\label{Sec:Proposed_analog}

\begin{figure}[t]
    \centering 
    \subfigure[End-to-end training for analog SC]{
    {\epsfig{file=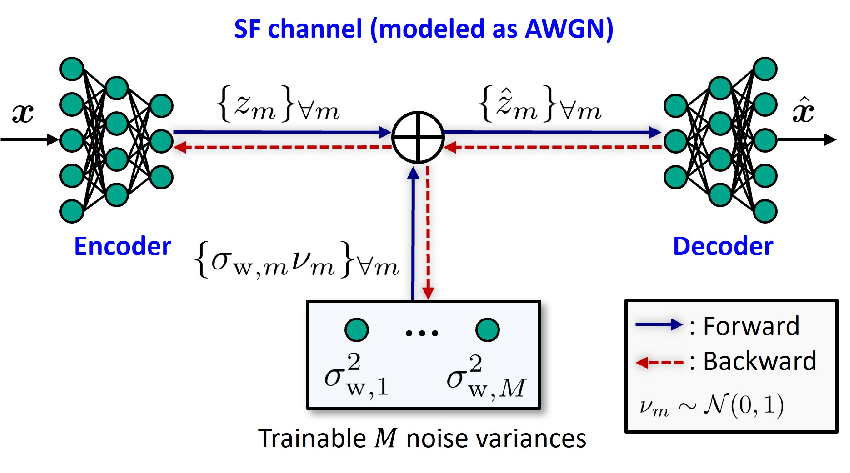, width=6.8cm}}
    }\vspace{1mm}
    \subfigure[End-to-end training for digital SC]{
    {\epsfig{file=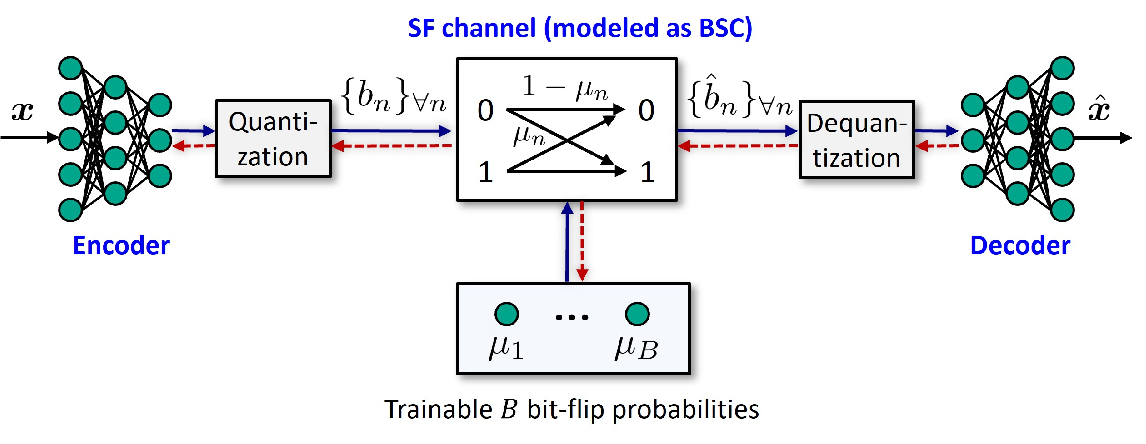, width=9cm}}
    }\\
    {\caption{The proposed end-to-end training strategy jointly optimizing the enc-dec and the SF channel for analog and digital SCs.}\label{fig:Train_SF}}\vspace{-3mm}
\end{figure}

For end-to-end training in analog SCs, we introduce a differentiable surrogate model based on an AWGN model, where $\hat{\bm z}={\bm z}+{\bm w}$ and ${\bm w}\sim \mathcal{N}({\bm 0},{\bm \Sigma}_{\rm ww})$. Motivated by the optimal SF channel structure derived in Sec.~III, the noise covariance ${\bm \Sigma}_{\rm ww}$ is treated as a trainable parameter, enabling the surrogate model to learn task-dependent unequal noise allocation for general SC scenarios with unknown source distributions and nonlinear enc-dec mappings. It is important to emphasize that this surrogate model is introduced solely to facilitate end-to-end optimization during training. Its realization in practical communication systems is presented in Sec.~\ref{Sec:PHY_Cali}, achieved through transceiver operations, such as fading channels, power allocation, and detection. 


The optimization problem for end-to-end training in analog SC is formulated as
\begin{align}
    ({\bf P1})~~&\min_{{\bm \theta}_{\rm enc},{\bm \theta}_{\rm dec},{\bm \Sigma}_{\rm ww}}  \mathbb{E}[\|{\bm x}-\hat{\bm x}\|^2], \label{eq:P5_obj}\\
    &~~~~~~~\text{s.t.}~~~~~~I({\bm z};\hat{\bm z})\leq I_{\rm max}. 
\end{align}
The primary challenge in solving ${\bf P1}$ is that the mutual information cannot be obtained in closed form, due to the nonlinear nature of DNN-based enc-dec and unknown input distributions. Moreover, even computing it numerically would incur high computational complexity, because the distributions of ${\bm z}$ and $\hat{\bm z}$ continuously shift with the enc-dec parameters ${\bm \theta}_{\rm enc},{\bm \theta}_{\rm dec}$, requiring $I({\bm z};\hat{\bm z})$ to be recomputed at every training iteration. 
To address this, we employ the mean-field assumption in \cite{RIB}, under which $I({\bm z};\hat{\bm z})$ is decomposed as follows: 
\begin{align}
    I({\bm z};\hat{\bm z}) = \sum_{m=1}^M I(z_m;\hat{z}_m). 
\end{align}
Based on this decomposition, an upper bound on the mutual information is given by 
\begin{align}
    I({\bm z};\hat{\bm z}) \leq \frac{1}{2} \sum_{m=1}^M \log \left( 1+\frac{\sigma_{{\rm z},m}^2}{\sigma_{{\rm w},m}^2} \right),
\end{align}
where $\sigma_{{\rm z},m}^2$ is the variance of $z_m$, which can be empirically estimated from training samples. From the above expression, we define the communication rate of the $m$-th SF as 
\begin{align}\label{eq:C_def_analog}
    C_m = \frac{1}{2} \log \left( 1+\frac{\sigma_{{\rm z},m}^2}{\sigma_{{\rm w},m}^2} \right),\quad \text{s.t.}~\sum_{m=1}^M C_m = I_{\rm max}.
\end{align}
To find the optimal $C_m$ via training, we parameterize it as 
\begin{align}\label{eq:rho}
    C_m = \rho_m I_{\rm max},\quad \text{s.t.}~\sum_{m=1}^M \rho_m = 1,
\end{align}
where $\rho_m \geq 0$ is a trainable parameter that determines the portion of the total rate assigned to the $m$-th SF. The constraint in \eqref{eq:rho} is directly derived from \eqref{eq:C_def_analog}. With this parameterization, problem ${\bf P1}$ is reformulated as a \textit{rate allocation problem} with optimization parameters $\{\rho_m\}_{\forall}$. 

The parameter $\rho_m$ can be readily implemented as 
\begin{align}\label{eq:comp_analog1}
\rho_m = \frac{|v_m|^2}{\sum_{i=1}^M |v_i|^2}, 
\end{align}
where $v_m\in\mathbb{R}$ denotes a trainable raw parameter. By the definition of $C_m$ in \eqref{eq:C_def_analog}, the noise variance is given by 
\begin{align}
    \sigma_{{\rm w},m}^2 = \frac{\sigma_{{\rm z},m}^2}{2^{2\rho_m I_{\rm max}}-1}.
\end{align}
Then, the training for the SF channel is realized as
\begin{align}\label{eq:comp_analog2}
    \hat{z}_m = z_m + \sigma_{{\rm w},m}\nu_m,~\nu_m\sim\mathcal{N}(0,1),~\forall m.
\end{align}
Here, the noise variance of the AWGN channel acts as a bias term in conventional DNNs. As a result, it can be readily optimized using standard neural network optimizers. 

\begin{figure}[t]
    \centering 
    \subfigure[Trained SF channel in analog SC]{
    {\epsfig{file=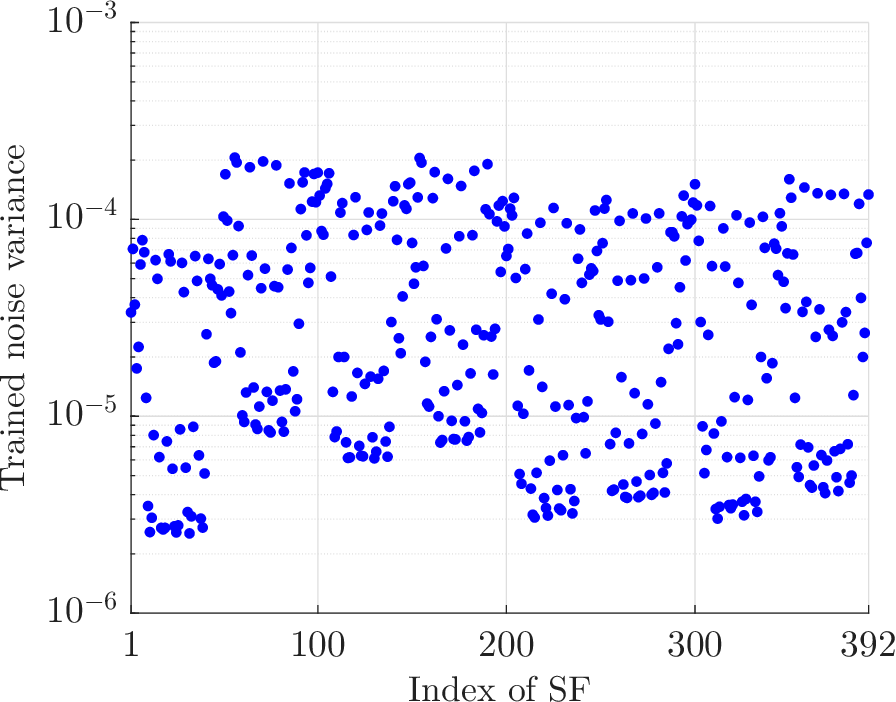, width=7cm}}
    }\vspace{0.01cm}
    \subfigure[Trained SF channel in digital SC]{
    {\epsfig{file=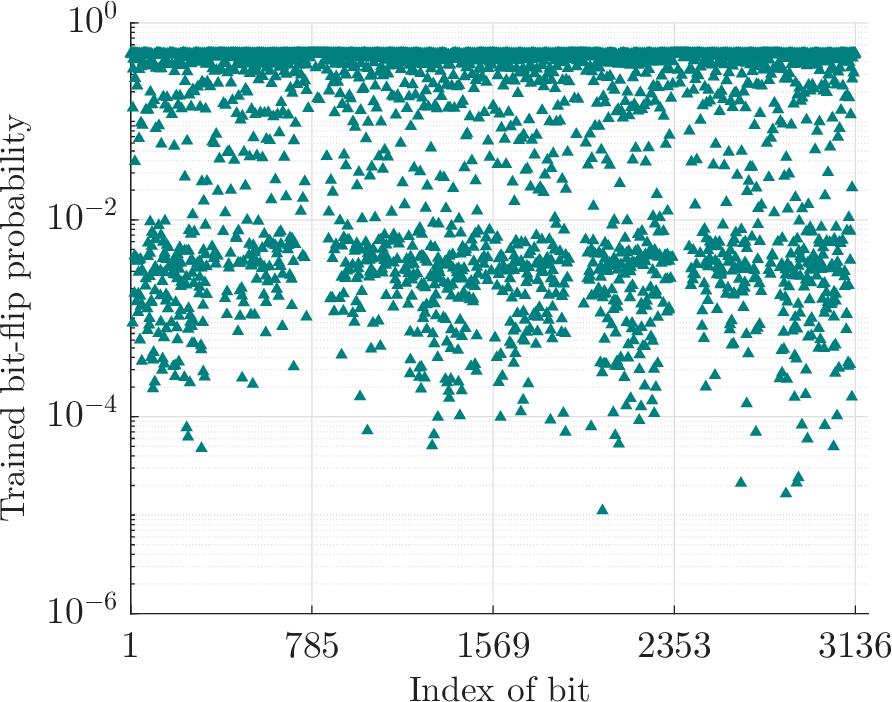, width=7cm}}
    }\\
    {\caption{Trained SF channel on the MNIST dataset with $I_{\rm max} = 1176$: (a) noise variance for analog SC, (b) bit-flip probability for digital SC.}\label{fig:Trained_SFC}}\vspace{-3mm}
\end{figure}

Fig.~\ref{fig:Trained_SFC}(a) presents a numerical example of the trained SF channel using the MNIST dataset with $I_{\rm max}=1176$. The detailed simulation settings are described in Sec.~\ref{Sec:Simul}. The result shows that the trained noise variance varies significantly across SFs. Such variation is not manually pre-configured but is obtained through end-to-end training. Recall that, in Fig.~\ref{fig:Gaussian_C_M}(a), SFs with a high source variance are assigned a low noise variance, which in turn leads to a low MSE. Similarly, in analog SCs, task-critical SFs are expected to be assigned a low noise variance to achieve a low task loss. In Sec.~\ref{Sec:Simul}, we will demonstrate that the proposed framework achieves superior task performance, which supports this interpretation. 


Meanwhile, in our training, only $M$ additional parameters $\{\rho_m\}_{\forall m}$ are introduced, which is sufficiently small compared to the number of enc-dec parameters. In addition, the associated computations in \eqref{eq:comp_analog1}--\eqref{eq:comp_analog2} are purely element-wise. Therefore, the proposed framework incurs only a marginal increase in computational complexity over conventional DeepJSCC.

\subsection{End-to-End Training for Digital SC}
For end-to-end training in digital SCs, we introduce a differentiable surrogate model based on relaxed BSCs with trainable bit-flip probabilities ${\bm \mu}$. In this model, the digital SF channel is given by
\begin{align}
    \hat{b}_n = \frac{(2b_n-1)\tilde{e}_n + 1}{2} \in [0,1],
\end{align}
where 
\begin{align}\label{eq:tilde_e}
    \tilde{e}_n = -\tanh\!\left(
        \frac{1}{\tau}\Big(
        \log\frac{\mu_n}{1-\mu_n} + \log\frac{u_n}{1-u_n}
        \Big)
    \right),
\end{align}
$u_n \sim \mathcal{U}(0,1)$ is a random variable, and $\tau$ is a temperature parameter \cite{BlindSC}. 
The relaxation is used to compute a gradient of $\mu_n$ with respect to a given loss function.

The optimization problem for the digital SC is formulated as
\begin{align}
    ({\bf P2})~~&\min_{{\bm \theta}_{\rm enc},{\bm \theta}_{\rm dec},{\bm \mu}}  \mathbb{E}[\|{\bm x}-\hat{\bm x}\|^2],\label{eq:P6_obj} \\
    &~~~~~\text{s.t.}~~~~I({\bm b};\hat{\bm b})\leq I_{\rm max}. \label{eq:P6_const_org}
\end{align}
The remaining procedures are similar to those in Sec.~\ref{Sec:Proposed_analog}. The mutual information is decomposed under the mean-field assumption, and an upper bound is obtained as
\begin{align}
    I({\bm b};\hat{\bm b}) \leq \sum_{n=1}^B \left(1-H_2(\mu_n)\right),
\end{align}
where $H_2(u)=-u\log_2 u -(1-u)\log_2(1-u)$ for $0\leq u\leq 0.5$. The communication rate of the $n$-th bit is defined as 
\begin{align}\label{eq:C_def_digital}
    C_n = 1-H_2(\mu_n),
\end{align}
subject to the following constraints: 
\begin{align}\label{eq:P6_const}
    \sum_{n=1}^B C_n = I_{\rm max},\quad 0\leq C_n\leq 1,
\end{align}
where the first constraint is derived from \eqref{eq:P6_const_org}, and the second constraint comes from $0\leq\mu_n \leq 0.5$. The rate allocation problem for digital SC is formulated by parameterizing
\begin{align}
    C_n = \rho_n I_{\rm max},~~\text{s.t. }\sum_{n=1}^B \rho_n = 1,~0\leq\rho_n\leq\frac{1}{I_{\rm max}}.
\end{align}
The parameter $\rho_n$ can be implemented as 
\begin{align}
    \rho_n = \frac{|v_n|^2 + \alpha}{\sum_{i=1}^B (|v_i|^2+\alpha)},
\end{align}
where $\alpha = \max\left(\frac{I_{\rm max}\max_i |v_i|^2 - \sum_i |v_i|^2}{B-I_{\rm max}},0\right)$ \cite{DLMPC}. From the definition of $C_n$ in \eqref{eq:C_def_digital}, the bit-flip probability of the $n$-th BSC is given by
\begin{align}
    \mu_n &= H_2^{-1}(1-\rho_nI_{\rm max}) \nonumber \\
    &\approx \frac{1}{2}-\sqrt{a\rho_nI_{\rm max} - b(\rho_nI_{\rm max})^2 - c(\rho_nI_{\rm max})^3 },
\end{align}
where $a=\frac{\log 2}{2}$, $b=\frac{(\log 2)^2}{6}$, and $c=a-b-\frac{1}{4}$. The approximation is used since $H_2^{-1}(\cdot)$ has no closed-form expression; it is obtained by performing a Taylor expansion of $H_2(u)$ around $u=0.5$, followed by series reversion.


Fig.~\ref{fig:Trained_SFC}(b) shows a numerical example of the trained bit-flip probability. The simulation settings are the same as those in Fig.~\ref{fig:Trained_SFC}(a), except that the digital SC is considered. The result shows that the trained bit-flip probability varies across bits, and as demonstrated in Sec.~\ref{Sec:Simul}, this leads to superior task performance. Meanwhile, similar to the analog case, digital SCs also introduce only $B$ additional parameters with purely element-wise computations. Therefore, the resulting increase in computational complexity is marginal. 


\vspace{1mm}
{\bf Remark 2 (Adaptation to Various Communication Environments):} 
Recall that, in Sec.~\ref{Sec:Comp}, we have discussed the trade-off between the mutual information limit $I_{\rm max}$ and the MSE. This naturally extends to SCs as a trade-off between $I_{\rm max}$ and the task performance, as demonstrated in Sec.~\ref{Sec:Simul}. To handle various communication environments, multiple enc-dec and SF channel pairs can be trained under different mutual information limits. 
In Sec.~\ref{Sec:PHY_Cali}, we introduce a communication strategy that adaptively selects an appropriate SF channel for a given communication environment. 

\vspace{1mm}
{\bf Remark 3 (Comparison to Prior Work in \cite{BlindSC}):} 
A similar approach was also studied in our prior work \cite{BlindSC}, where a BSC-based SF channel was optimized via end-to-end training. However, the optimization relied on a heuristically designed loss function, rather than capturing or constraining the mutual information of the SF channel. 
Moreover, its validation was restricted to digital SC, raising concerns about its scalability to other forms of SC scenarios, e.g., analog SC. The advantage of our mutual-information-constrained approach over the heuristic approach in \cite{BlindSC} will be further discussed in Sec.~\ref{Sec:Simul}.

\section{Proposed PHY Calibration for Realizing the Trained SF Channel}\label{Sec:PHY_Cali}
The training framework in Sec.~\ref{Sec:Train} produces the optimized SF channel by imposing a mutual information constraint, which captures the effects of various communication constraints in an integrated manner during training. However, this abstraction does not directly guarantee that the trained SF channel can be realized under practical communication settings because it does not explicitly account for communication constraints such as total transmit power. To address this issue, the communication parameters should be calibrated so that the actual SF channel observed during communication aligns with the trained one while satisfying communication constraints. We refer to this process as \textit{PHY calibration}. In this section, we present PHY calibration strategies for two communication settings: (i) single-user analog SCs and (ii) multi-user digital SCs.

\subsection{Single-User Analog SCs}\label{Sec:PHY_Cali_A}
Consider the SF channels trained for different mutual information limits $\{I_{\rm max}^{(u)}\}_{u=1}^U$, satisfying $I_{\rm max}^{(1)} > \cdots > I_{\rm max}^{(U)}$, as discussed in {\bf Remark 2}. The corresponding losses $\{L^{(u)}\}_{u=1}^U$ follow $L^{(1)} < \cdots < L^{(U)}$. Our objective for PHY calibration is to jointly select a proper SF channel and the transmit power. The optimization problem is formulated as
\begin{align}
    ({\bf P3})~&\min_{\{{p}_m\}_{\forall m},\,u}~~ (1-w_0)L^{(u)} + w_0 \sum_{m=1}^{M} {p}_m \sigma_{{\rm z},m}^2, \\
    \text{s.t.}~~ &\frac{2|h_t|^2 {p}_m\sigma_{{\rm z},m}^2}{ \sigma^2 } \geq \overline{\rm SNR}_{m}^{(u)},~\forall m,t, \label{eq:P7_const1}~\sum_{m=1}^{M} {p}_m\sigma_{{\rm z},m}^2 \leq P_{\rm tot}, 
\end{align}
where $p_m\sigma_{{\rm z},m}^2$ represents the average transmit power used for sending the $m$-th SF, and $w_0 \in [0,1]$ controls the trade-off between the task loss and the total transmit power. The target SNR of the $m$-th SF in the $u$-th SF channel, denoted by $\overline{\rm SNR}_{m}^{(u)}$, is defined as 
\begin{align}
    \overline{\rm SNR}_{m }^{(u)}\triangleq \frac{\sigma_{{\rm z},m}^2}{(\sigma_{{\rm w},m}^{(u)})^2},
\end{align}
where $(\sigma_{{\rm w},m}^{(u)})^2$ is the trained noise variance of the $m$-th SF in the $u$-th SF channel. In the first constraint, $\frac{2|h_t|^2 p_m \sigma_{{\rm z},m}^2}{\sigma^2}$ represents the actual SNR of $z_m$ during transmission. This constraint ensures alignment between the target and actual SNRs, thereby improving the reliability of task performance.

To solve problem ${\bf P3}$, an auxiliary variable is precomputed as
\begin{align}
    \tau_m^{(u)} = \frac{\overline{\rm SNR}_{m}^{(u)}}{2\sigma_{{\rm z},m}^2},~ \forall m,u.
\end{align}
For each $u$, $\tau_m^{(u)}$ is sorted in descending order with respect to $m$ in advance. When communication begins, the channel-gain-to-noise-power ratio $\tfrac{|h_t|^2}{\sigma^2}$ is also sorted in descending order. Here, the indices $m$ and $t$ are retained after sorting for notational simplicity. The required power coefficient is then computed as
\begin{align}\label{eq:PHY_analog_p_m_u}
    \bar{p}_m^{(u)} = \frac{\tau_m^{(u)}\sigma^2}{|h_t|^2},\quad t=\frac{\lceil m \rceil}{2}.
\end{align}
The sorting above assigns SFs with higher $\tau_m^{(u)}$ to stronger channels, thereby reducing the total transmit power.
After obtaining $\bar{p}_m^{(u)}$, the optimal SF channel index is determined as 
\begin{align}\label{eq:PHY_analog_u_star}
    u^\star = \argmin_{u} \left(1-w_0)L^{(u)} + w_0 P_{\rm req}^{(u)}: P_{\rm req}^{(u)} \leq P_{\rm tot}\right),
\end{align}
where
\begin{align}
    P_{\rm req}^{(u)} = \sum_{m=1}^{M} \bar{p}_m^{(u)}\sigma_{{\rm z},m}^2.
\end{align}
The optimal power coefficient is given by $\bar{p}_m^{(u^\star)}$. 

The proposed PHY calibration strategy for analog SC has three advantages. 
First, since $\tau_m^{(u)}$ can be pre-shared between the transmitter and the receiver, the optimal power can be computed locally once $\tfrac{|h_t|^2}{\sigma^2}$ is obtained, requiring no additional communication overhead. Second, the strategy incurs low computational complexity, as the optimal power is obtained in closed form with simple arithmetic operations. 
Finally, the strategy can be readily extended to an interference-free multi-user scenario, where each user independently adjusts its transmit power, without requiring coordination among users. Meanwhile, in our strategy, we have considered only a total power constraint. Nevertheless, the strategy can be extended to various practical constraints; as a representative example, in Sec.~\ref{Sec:Simul}, we further impose a per-time-slot power constraint to reflect practical hardware limits on instantaneous transmit power.

\subsection{Multi-User Digital SCs}\label{Sec:Digital_Target_Realize}
We consider a multi-user digital SC where $K$ users transmit different images to a single base station (BS). The channels of all users are assumed to be independent and remain constant during the transmission of all symbols.
For the $k$-th user, consider SF channels trained for different mutual information limits $\{I_{{\rm max},k}^{(u_k)}\}_{u_k=1}^{U_k}$, satisfying $I_{{\rm max},k}^{(1)} > \cdots > I_{{\rm max},k}^{(U_k)}$. The corresponding losses $\{L_k^{(u_k)}\}_{u_k=1}^{U_k}$ follow $L_k^{(1)} < \cdots < L_k^{(U_k)}$. Our objective is to jointly determine a proper SF channel, the transmit power, and the modulation levels. The optimization problem is formulated as 
\begin{align}
    &({\bf P4})~~\min_{  \big\{\{p_{t,k}\}_{\forall t}, u_k, m_k\big\}_{\forall k}  }~\sum_{k=1}^K w_k L_k^{(u_k)} + w_0 \sum_{k=1}^K\sum_{t=1}^{T_k} p_{t,k} \\
    &\text{s.t.}~\bar{\mu}_{n,k}^{(u_k)} \geq {\rm BER}\left(p_{t,k},m_k,\frac{|h_{k}|^2}{\sigma^2}\right),~\forall k,n\in\{1,\cdots,B_k\}, \\ 
    &~\sum_{t=1}^{T_k} p_{t,k} \leq P_{{\rm tot}}^{(k)},~\forall k,~\sum_{k=1}^K T_k \leq T,~m_k\in\{2,4,6,\cdots\},~\forall k,\label{eq:P6_const_MU}
\end{align}
where $p_{t,k}$ is the transmit power for the $t$-th symbol, $m_k$ is the modulation level, $B_k$ is the number of transmitted bits, and $T_k = B_k/m_k$ is the corresponding symbol sequence length for the $k$-th user. The weighting factors $w_0$ and $w_k$ control the trade-off between the total power consumption and the task performance of each user.
In the first constraint, $\bar{\mu}_{n,k}^{(u_k)}$ denotes the trained bit-flip probability of the $n$-th bit in the $u_k$-th SF channel. Each $n$-th bit is transmitted within the $t$-th symbol, where $t=\lceil{ n }/{m_k}\rceil$. 
The BER for this bit is defined as
\begin{align}
    {\rm BER}&\left(p_{t,k},m_k,\frac{|h_{k}|^2}{\sigma^2}\right) \triangleq a(m_k) {\rm erfc} \Bigg( \sqrt{\frac{c(m_k){p_{t,k}|h_{k}|^2}}{\sigma^2}} \Bigg)\nonumber \\&~~\qquad\qquad + b(m_k) {\rm erfc} \Bigg( 3\sqrt{\frac{c(m_k){ p_{t,k}|h_{k}|^2}}{\sigma^2}} \Bigg),
\end{align}
where $h_{k}\in\mathbb{C}$ is the channel coefficient of the $k$-th user, $a(m_k)=\frac{\sqrt{2^{m_k}} - 1}{\sqrt{2^{m_k}} \log_2 \sqrt{2^{m_k}}}$, $b(m_k)=\frac{\sqrt{2^{m_k}} - 2}{\sqrt{2^{m_k}} \log_2 \sqrt{2^{m_k}}}$, and $c(m_k)=\frac{3}{2(2^{m_k} - 1)}$ \cite{BER_example}. The second constraint limits the total power budget of each user. The third constraint guarantees that the total number of channel uses across all users does not exceed $T$, and the fourth constraint is the candidate modulation levels.

To solve problem ${\bf P4}$, we first sort $\bar{\mu}_{n,k}^{(u_k)}$ in descending order with respect to $n$ in advance, where the index $n$ is retained for notational simplicity. 
The sorted bit-flip probabilities are grouped by every $m_k$ bits, and the minimum value within each group is defined as
\begin{align}
    \bar{\mu}_{t,k}^{(u_k,m_k)} = \min_{n\in\{(t-1)m_k+1,\cdots, tm_k\}}\big\{\bar{\mu}_{n,k}^{(u_k)}\big\},
\end{align}
for $t\in\{1,\cdots,t_k(m_k)\}$, where $t_k(m_k)\triangleq B_k/m_k$. 
The sorting above groups bits with similar bit-flip probabilities. This helps reduce the total transmit power because the transmit power of each symbol is determined by the minimum bit-flip probability within its group, as described in below. 
Given $m_k$ and $\bar{\mu}_{t,k}^{(u_k,m_k)}$, an auxiliary variable is precomputed as 
\begin{align}
    \gamma_{t,k}^{(u_k,m_k)} = \min\{p : \bar{\mu}_{t,k}^{(u_k,m_k)} \geq {\rm BER}(p,m_k,1)\},
\end{align}
for all $(t,k,u_k,m_k)$, assuming $|h_k|^2/\sigma^2=1$. When communication begins, the actual channel-gain-to-noise-power ratio $\frac{|h_k|^2}{\sigma^2}$ is used to determine the transmit power as 
\begin{align}
    \bar{p}_{t,k}^{(u_k,m_k)} = \frac{\gamma_{t,k}^{(u_k,m_k)}\sigma^2}{|h_k|^2}. 
\end{align} 
Under the total power constraint, the feasible set for the $k$-th user is defined as
\begin{align}
    \Omega_k = \{(u_k,m_k): P_{{\rm req},k}^{(u_k,m_k)} \leq P_{\rm tot}^{(k)}\},
\end{align}
where $P_{{\rm req}, k}^{(u_k,m_k)} = \sum_{t=1}^{T_k} \bar{p}_{t,k}^{(u_k,m_k)}$. 
For each feasible pair $(u_k,m_k)\in\Omega_k$, the corresponding objective value is given by
\begin{align}
    J_k(u_k,m_k) = w_k L_k^{(u_k)} + w_0 P_{{\rm req},k}^{(u_k,m_k)}.  
\end{align}

For notational convenience, we redefine
\begin{align}
    t_{k,j} = t_k(m_k), \quad J_{k,j} = J_k(u_k,m_k),
\end{align}
where $j\in\{1,\dots,|\Omega_k|\}$ indexes each feasible pair $(u_k,m_k)\in\Omega_k$. Then, problem ${\bf P4}$ is reformulated as 
\begin{align}
    &({\bf P4}^\prime)~~\min_{  \{x_{k,j}\}_{\forall k,j}  }~\sum_{k=1}^K\sum_{j=1}^{|\Omega_k|} J_{k,j}x_{k,j} \\
    &\text{s.t. }\sum_{j=1}^{|\Omega_k|} x_{k,j}=1,~x_{k,j}\in\{0,1\}, \sum_{k=1}^K\sum_{j=1}^{|\Omega_k|} t_{k,j}x_{k,j} \leq T,
\end{align}
where the first two constraints ensure that exactly one candidate is selected from the feasible set $\Omega_k$ for the $k$-th user. The third constraint corresponds to the total channel-use constraint in \eqref{eq:P6_const_MU}. We note that problem ${\bf P4}^\prime$ is a conventional multiple-choice knapsack problem. 
This is a well-studied combinatorial optimization problem, and many efficient solvers have been developed \cite{MCKP_book}. From a computational complexity perspective, the worst-case approach is exhaustive search, which evaluates all combinations across the candidate sets for each user. In practice, however, the number of candidate SF channels and modulation levels per user is small, resulting in moderate computational cost.

In the proposed PHY calibration for multi-user digital SC, the optimal SF channel index and modulation level $(u_k^\star, m_k^\star)$ are first determined at the BS by solving ${\bf P4}^\prime$. The BS then transmits $(u_k^\star, m_k^\star)$ and $\frac{|h_k|^2}{\sigma^2}$ to each user. Upon receiving them, each user computes the optimal transmit power as $\bar{p}_{t,k}^\star=\frac{\gamma_{t,k}^{(u_k^\star,m_k^\star)}\sigma^2}{|h_k|^2}$, which can also be computed at the BS. Therefore, only a small amount of information needs to be exchanged.


\begin{figure}[t]
    \centering 
    \subfigure[PSNR vs. $I_{\rm max}$ (Analog SC)]{
    {\epsfig{file=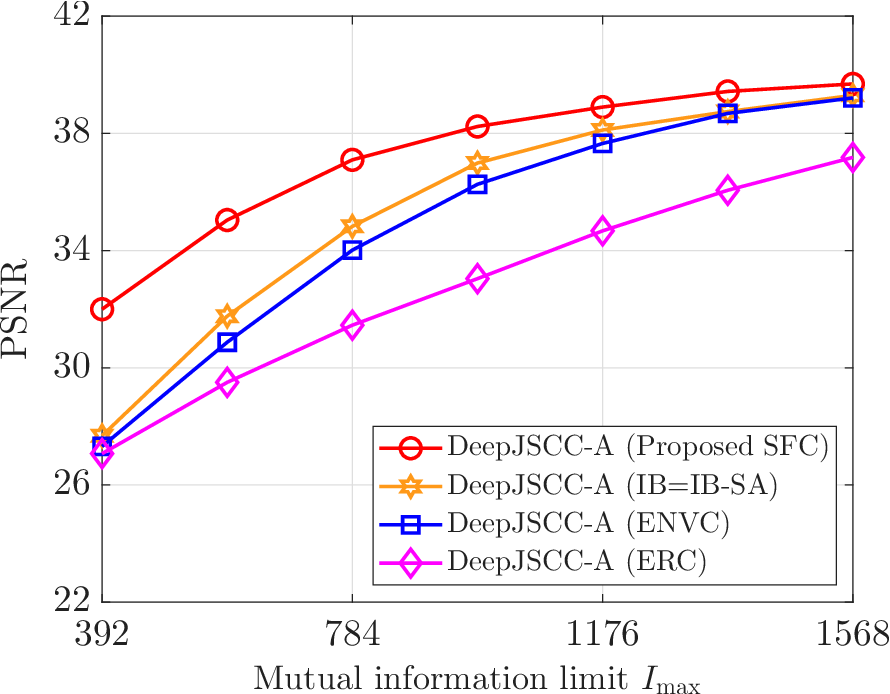, width=7cm}}
    }\hspace{-0.7cm}
    \subfigure[PSNR vs. $M$ (Analog SC)]{
    {\epsfig{file=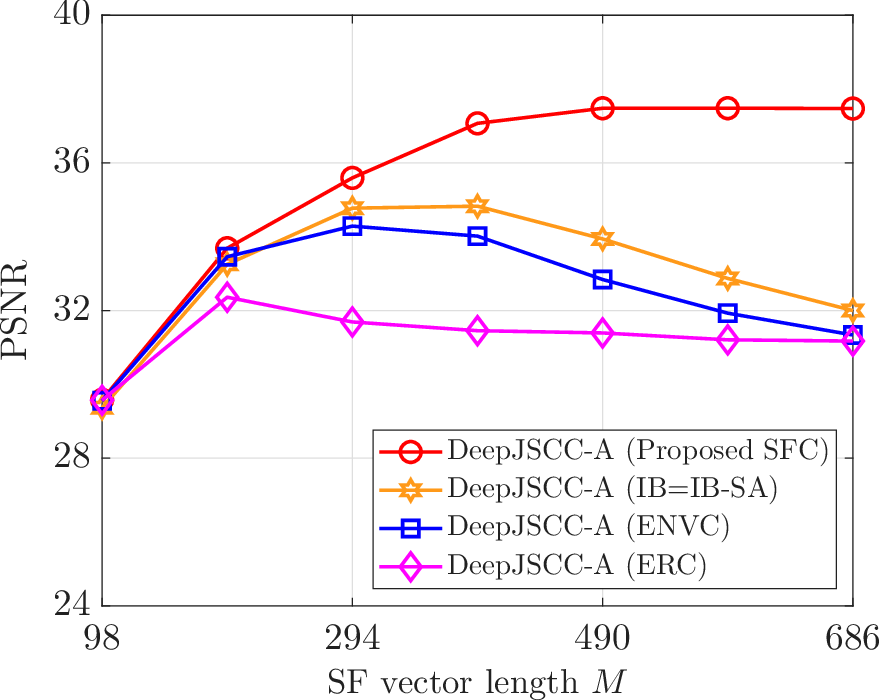, width=7cm}}
    }\\
    {\caption{PSNR performance over the mutual information limit $I_{\rm max}$ and the SF vector length $M$ for analog SCs on the MNIST dataset.}\label{fig:SC_A_C_M}}\vspace{-3mm}
\end{figure}

\section{Simulation Results}\label{Sec:Simul}
In this section, we demonstrate the superiority of the proposed SF channel in SCs, using the MNIST \cite{MNIST}, CIFAR-$10$ \cite{CIFAR10}, and STL-$10$ \cite{STL10} datasets. The enc-dec architecture follows either a Swin Transformer-based SwinJSCC in \cite{SwinJSCC}, or the same configuration as in \cite{BlindSC} with the activation function of the last encoder layer replaced with a sigmoid. 
The loss function is used as the MSE loss when evaluating with the PSNR, and the SSIM loss when evaluating with the SSIM \cite{DeepJSCC_Q}. 
For MNIST and CIFAR-$10$, the number of training epochs is set to 50 for PSNR and 20 for SSIM, while 100 epochs are used for STL-$10$. 
The batch size is fixed to $64$ for all datasets, and the Adam optimizer is employed with an initial learning rate of $10^{-4}$. 


For performance comparison of analog SCs, we consider the following baselines.
\begin{itemize}
    \item {\bf DeepJSCC-A (Proposed SFC):} This framework integrates the proposed SF channel (SFC) optimization into the analog DeepJSCC framework of \cite{DeepJSCC}. 
    

    \item {\bf DeepJSCC-A (ENVC) \cite{DeepJSCC}:} This framework corresponds to the original analog DeepJSCC of \cite{DeepJSCC} without any SFC optimization. The SF channel is modeled as an equal-noise-variance channel (ENVC), in which all SFs are corrupted by Gaussian noise with the same variance. 

    \item {\bf DeepJSCC-A (ERC):} This variant modifies the conventional DeepJSCC by explicitly imposing an equal-rate constraint across all SFs. Specifically, the noise variance of the $m$-th SF is adjusted so that its communication rate satisfies $C_m = \frac{I_{{\rm max}}}{M}$. 
    
    \item {\bf DeepJSCC-A (IB) \cite{IB} / DeepJSCC-A (IB-SA) \cite{IB_SA}:} Both baselines train the encoder–decoder using an information bottleneck (IB)-based loss function. During communication, the baseline in [34] evaluates the robustness of SFs to noise and allocates SFs with lower robustness to stronger subchannels.
\end{itemize}

For performance comparison of digital SCs, we consider the following baselines.
\begin{itemize}
    \item {\bf DeepJSCC-D (Proposed SFC):} This framework incorporates the proposed joint optimization into the digital DeepJSCC of \cite{NECST}. 

    \item {\bf DeepJSCC-D (ENVC = ERC) \cite{NECST}:} This framework can be regarded as a quantized version of DeepJSCC-A (ENVC), extending the one-bit quantization process in \cite{NECST} to a multi-bit representation. For training, it adopts multiple BSCs with an equal bit-flip probability applied to all bits, resulting in equal rate allocation. 

    \item {\bf DeepJSCC-D (BlindSC) \cite{BlindSC}:} This framework corresponds to the digital SC framework in \cite{BlindSC}. All bit-flip probabilities are initialized equally to satisfy the mutual information limit $I_{\rm max}$, and the regularization weight is tuned so that the constraint is maintained at the final training epoch.


\end{itemize}
All digital SC frameworks use an 8-bit uniform quantizer for the encoder output. 

\vspace{1mm}

\begin{figure}[t]
    \centering 
    \subfigure[PSNR vs. $I_{\rm max}$ (Digital SC)]{
    {\epsfig{file=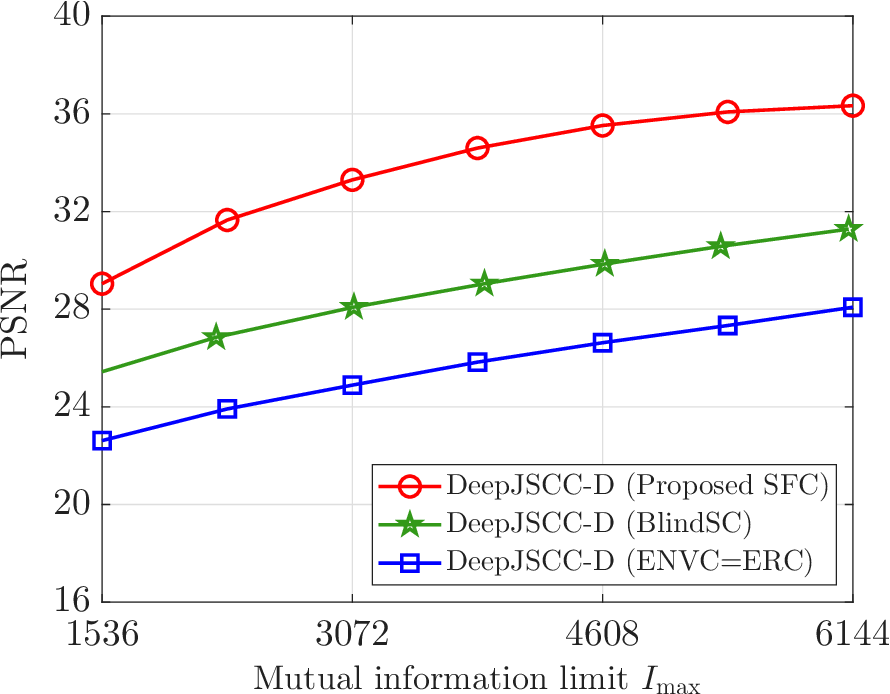, width=7cm}}
    }\hspace{-0.7cm}
    \subfigure[PSNR vs. $B$ (Digital SC)]{
    {\epsfig{file=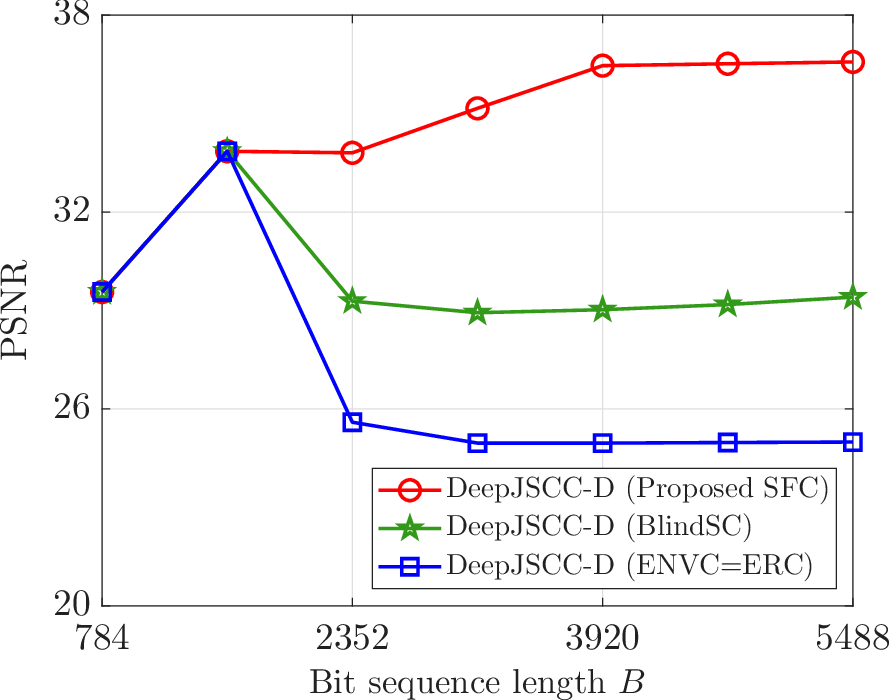, width=7cm}}
    }\\
    {\caption{PSNR performance over the mutual information limit $I_{\rm max}$ and the SF vector length $B$ for digital SCs on the CIFAR-$10$ dataset.}\label{fig:SC_D_C_M}}\vspace{-3mm}
\end{figure}

Fig.~\ref{fig:SC_A_C_M} shows the PSNR performance of analog SCs on the MNIST dataset for different values of the mutual information limit $I_{\rm max}$ and the SF vector length $M$. In Fig.~\ref{fig:SC_A_C_M}(a), $M$ is fixed to $392$ (corresponding to $N/M=2$), while in Fig.~\ref{fig:SC_A_C_M}(b), $I_{\rm max}$ is fixed to $784$. Similar to the Gaussian case, Fig.~\ref{fig:SC_A_C_M}(a) shows that the proposed SFC consistently achieves the highest PSNR across all values of $I_{\rm max}$. This indicates that the proposed SFC utilizes the available mutual information more effectively than the baselines by optimizing the SF channel. In Fig.~\ref{fig:SC_A_C_M}(b), when $M$ is small, all schemes yield relatively low PSNR due to strong compression. However, as $M$ increases, the PSNR of the proposed SFC gradually improves and eventually converges. This is because a larger $M$ preserves more information from the input data, but the gains diminish due to the limited mutual information. In contrast, the ENVC, ERC, and IB baselines initially show an increase in PSNR but begin to degrade as $M$ becomes large. This degradation occurs because increasing $M$ forces stronger noise to be assigned to all SFs, thereby distorting even the task-critical SFs.

Fig.~\ref{fig:SC_D_C_M} shows the PSNR performance of digital SCs on the CIFAR-$10$ dataset for different values of the mutual information limit $I_{\rm max}$ and the bit sequence length $B$. The enc-dec architecture follows a Swin Transformer-based SwinJSCC in \cite{SwinJSCC}. In Fig.~\ref{fig:SC_D_C_M}(a), $B$ is fixed to $12,288$ (corresponding to $8N/B=2$), while in Fig.~\ref{fig:SC_D_C_M}(b), $I_{\rm max}$ is fixed to $1,568$. Fig.~\ref{fig:SC_D_C_M}(a) shows that the proposed SFC consistently outperforms the other baselines over the entire range of $I_{\rm max}$. In Fig.~\ref{fig:SC_D_C_M}(b), when $B\leq1,568$, the bit sequence length $B$ is smaller than or equal to $I_{\rm max}$. 
In this case, the communication becomes error-free, and all schemes achieve identical PSNR values. 
Meanwhile, the comparison with BlindSC demonstrates that the proposed SFC achieves superior performance by leveraging an information-theoretic optimization instead of heuristic loss design.

\begin{figure}[t]
    \centering 
    \subfigure[PSNR vs. SNR (Analog SC)]{
    {\epsfig{file=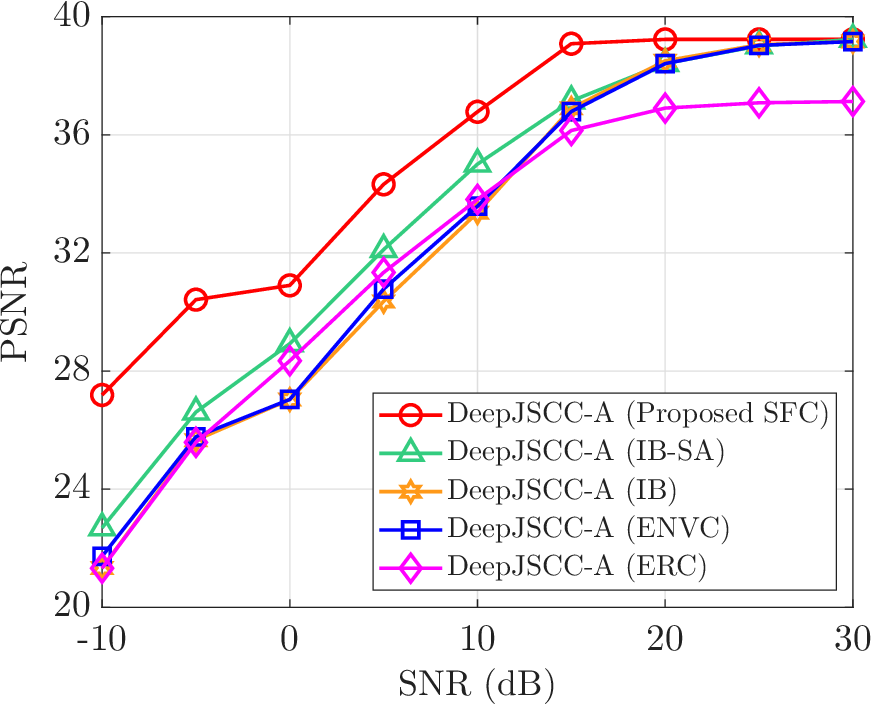, width=7cm}}
    }\hspace{-0.7cm}
    \subfigure[PSNR vs. SNR (Digital SC)]{
    {\epsfig{file=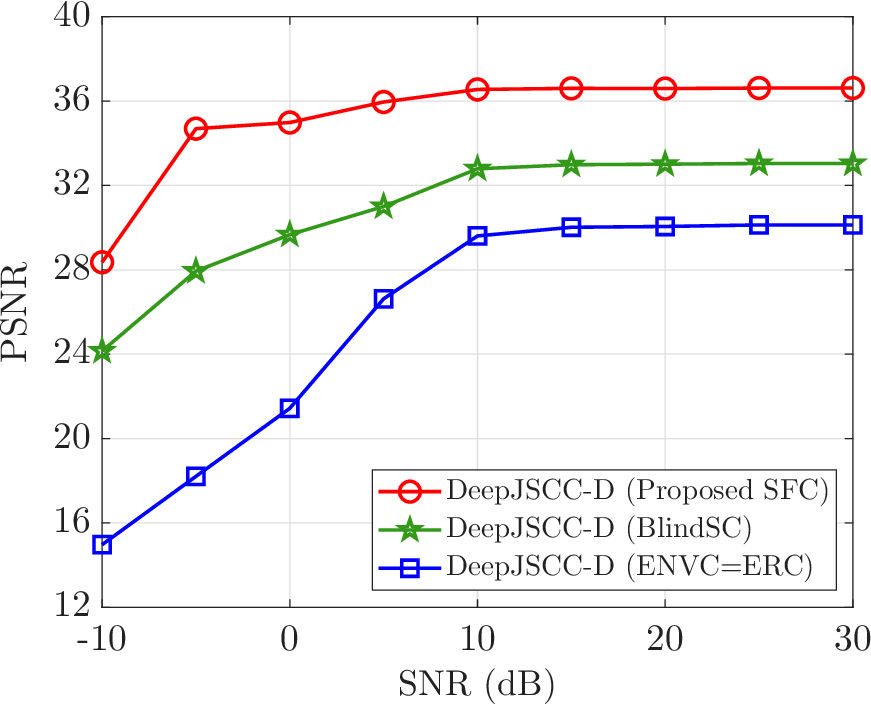, width=7cm}}
    }\\
    {\caption{PSNR performance over the SNR for single-user analog and digital SCs on the MNIST dataset.}\label{fig:MNIST_SNR}}\vspace{-3mm}
\end{figure}

\begin{figure*}[t]
    \centering 
    \subfigure[User 1 (MNIST)]{
    {\epsfig{file=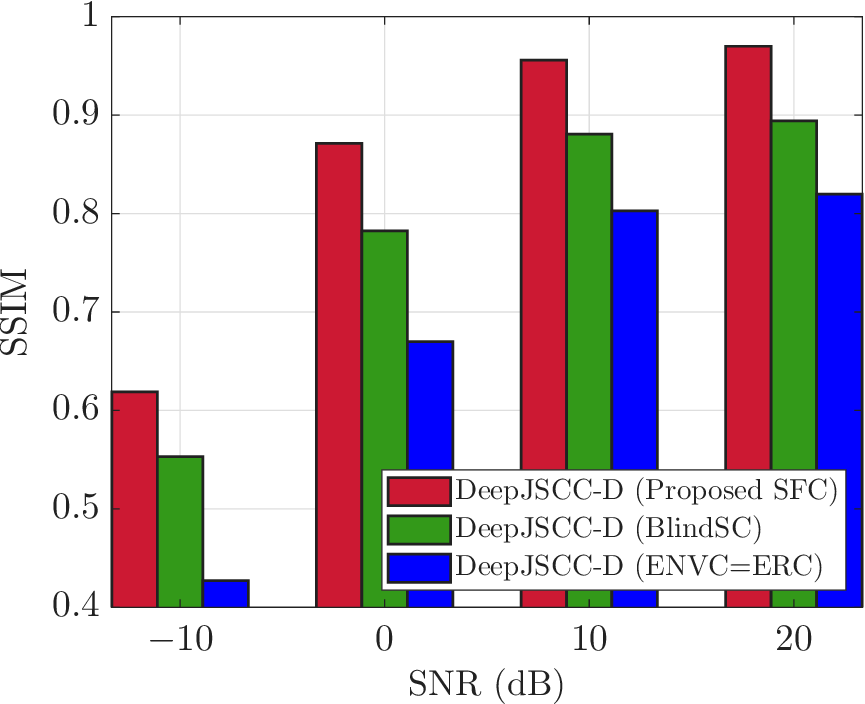, width=5.8cm}}
    }\hspace{0.01cm}
    \subfigure[User 2 (CIFAR-$10$)]{
    {\epsfig{file=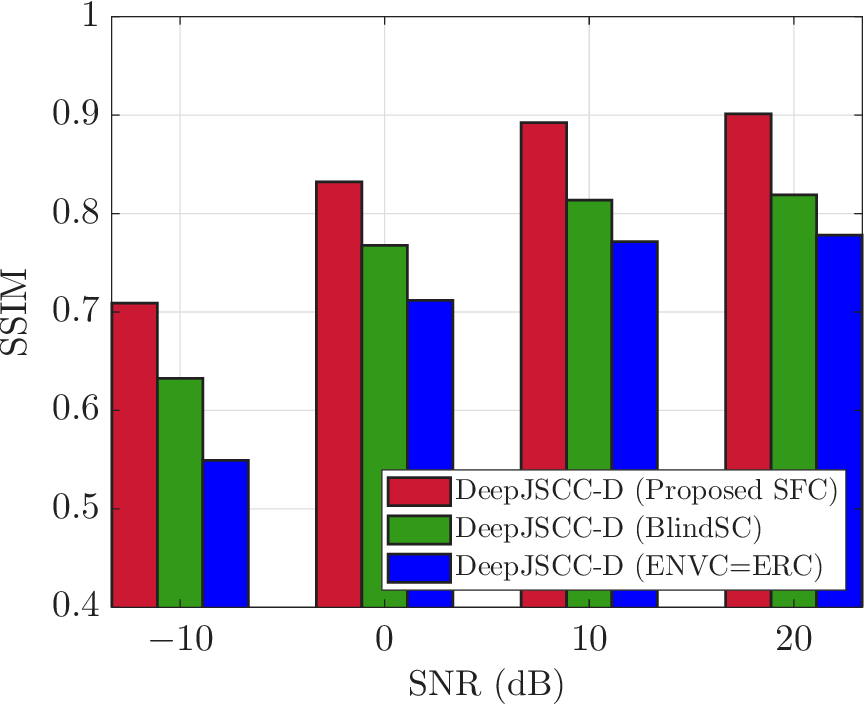, width=5.8cm}}
    }\hspace{0.01cm}
    \subfigure[User 3 (STL-$10$)]{
    {\epsfig{file=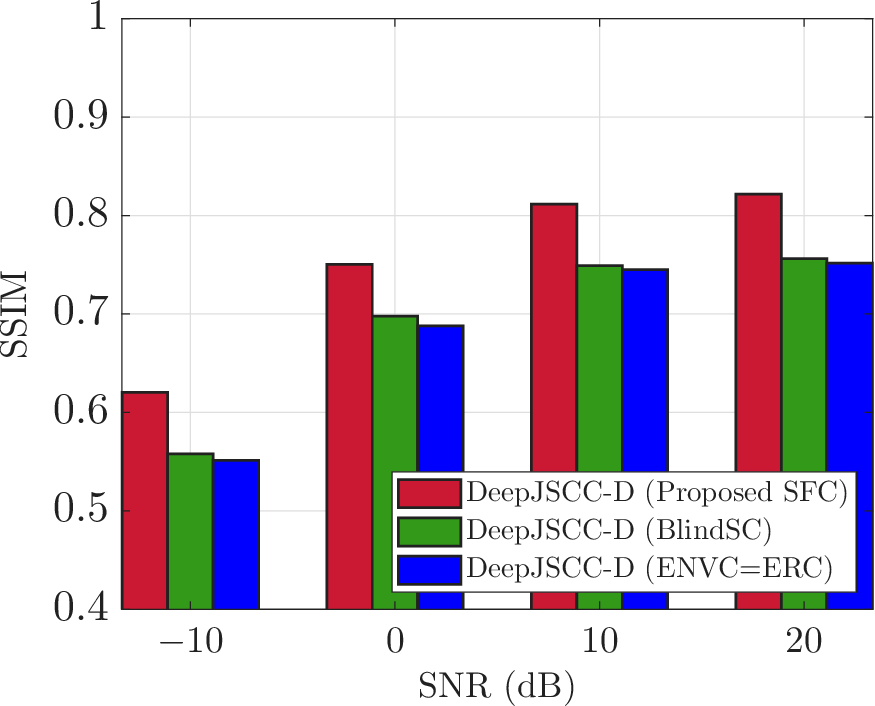, width=5.8cm}}
    }\\
    {\caption{SSIM performance over the SNR for multi-user digital SCs on MNIST, CIFAR-$10$, and STL-$10$ datasets.}\label{fig:MU}}\vspace{-3mm}
\end{figure*}

Fig.~\ref{fig:MNIST_SNR} shows the PSNR performance of single-user analog and digital SCs on the MNIST dataset for different values of SNR. In this simulation, we set $P_{\rm tot}=10^4$ and $w_0 \ll 1$. For analog SC, $I_{{\rm max}}^{(u)}=392u,~u\in\{1,2,3,4\}$. For digital SC, $I_{{\rm max}}^{(v)}=392v,~v\in\{1,3,5,7\}$, and 4-QAM is used. For both SCs, the transmission is performed over $\lceil T/10 \rceil$ Rayleigh fading subchannels, each spanning 10 channel uses. For fair comparison, all schemes, except for IB-SA, follow the PHY calibration strategy in Sec.~\ref{Sec:PHY_Cali_A} with their respective target SNRs or BERs. For IB-SA, since there is no criterion to select the enc-dec pair for a given SNR, we evaluate multiple enc-dec pairs and report the best performance at each SNR. The results show that the proposed SFC consistently achieves the highest PSNR across all SNR regimes. 
Notably, the performance trend observed here aligns well with Figs.~\ref{fig:SC_A_C_M} and \ref{fig:SC_D_C_M}. 
This consistency demonstrates that the optimized SF channel trained under the mutual information constraint can be realized in practical wireless environments through the proposed PHY calibration strategy. In other words, even though the training of the SF channel is performed in an abstract mutual-information domain, its performance advantage seamlessly transfers to real physical channels once the PHY calibration is applied.

\begin{figure}[t]
    \centering
    {\epsfig{file=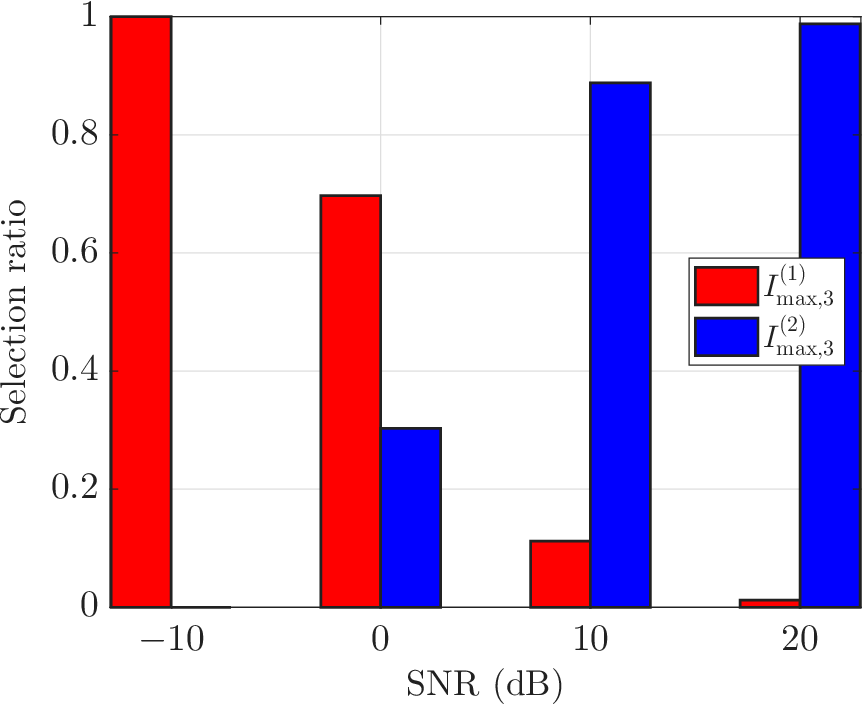, width=7cm}}
    \caption{Selection ratios over the SNR for the user transmitting the STL-$10$ dataset in multi-user digital SCs.}\label{fig:Selection}
\end{figure}

Fig.~\ref{fig:MU} shows the SSIM performance of multi-user digital SCs for different values of SNR. In this simulation, we consider three users, where each user transmits images from a different dataset (MNIST, CIFAR-10, and STL-10). 
For each dataset, the SF vector length $M$ is chosen such that $N/M=8$ holds. The mutual information limits are set as $I_{{\rm max},k}^{(1)}=B_k/8$ and $I_{{\rm max},k}^{(2)}=B_k/2$ for all $k$, while the total transmit powers for the three users are set to $10^3$, $10^4$, $10^5$, respectively. Each user experiences an independent Rayleigh fading channel. 
The other parameters are set as $T=10^4$, $w_0 \ll 1$, and $w_k=1, \forall k$.
For fair comparison, all schemes follow the PHY calibration strategy in Sec.~\ref{Sec:Digital_Target_Realize} with their respective target bit-flip probabilities, and the problem ${\bf P4}^\prime$ is solved using full search. 
The results show that the proposed SFC consistently achieves the highest SSIM across all SNR values and datasets. 

Fig.~\ref{fig:Selection} shows the selection ratios of $I_{{\rm max},3}^{(1)}$ and $I_{{\rm max},3}^{(2)}$ over the SNR for the third user transmitting the STL-$10$
dataset. The simulation setting is the same as that in Fig.~\ref{fig:MU}. The result shows that the user mainly selects $I_{{\rm max},3}^{(1)}$ when the SNR is low and switches to $I_{{\rm max},3}^{(2)}$ as the SNR increases. Similar trends can also be found for the other users, which are omitted for brevity. Overall, the result demonstrates that the proposed PHY calibration strategy adaptively chooses the appropriate enc-dec pair depending on the channel condition. 

\begin{figure}[t]
        \centering
        {\epsfig{file=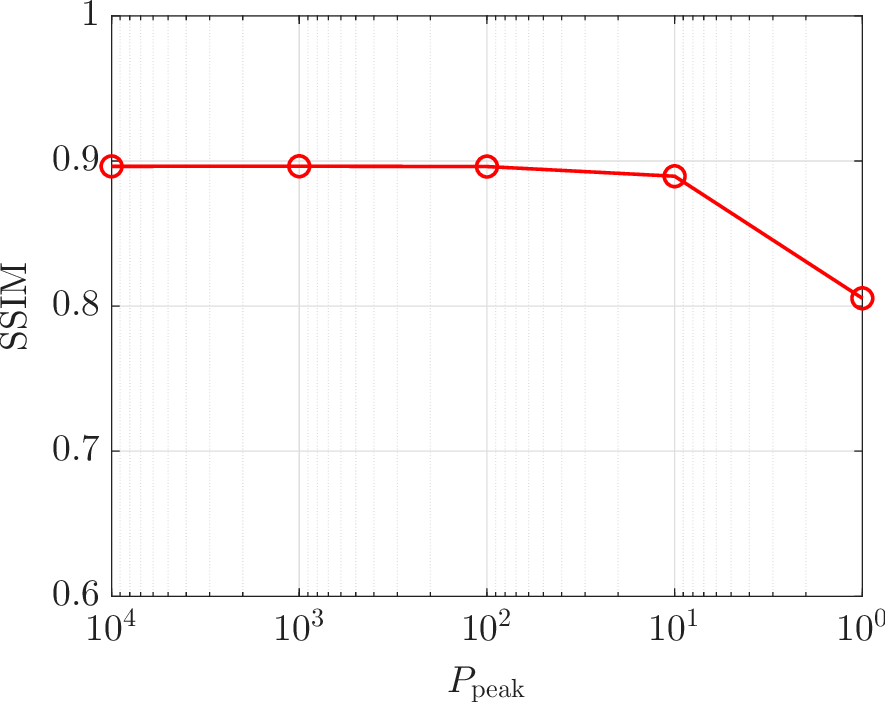, width=7cm}}
        \caption{SSIM performance over $P_{\rm peak}$ under a per-time-slot power constraint for the user transmitting the CIFAR-10 dataset in multi-user digital SCs.}  
        \label{fig:power}\vspace{-3mm}
    \end{figure}

Fig.~\ref{fig:power} shows the SSIM performance on the CIFAR-$10$ dataset under a per-time-slot power constraint when SNR = $10$ dB. The other simulation settings are the same as those in Fig.~\ref{fig:MU}. The per-time-slot power constraint is realized by clipping the transmit power $p_{t,k}$ obtained from solving Problem ${\bf P4}$ to a predefined maximum value $P_{\rm peak}$, i.e., $p_{t,k}\leftarrow \min\{p_{t,k},P_{\rm peak}\}$, where $P_{\rm peak}=P_{\rm tot}/u$ for $u\in\{1,10,10^2,10^3,10^4\}$. Fig.~\ref{fig:power} shows that the SSIM performance remains nearly unchanged for moderate values of $P_{\rm peak}$, while a noticeable degradation appears only when $P_{\rm peak}$ becomes overly restrictive (i.e., $P_{\rm peak}=1$). This shows that the framework extends to a per-time-slot power constraint via post-processing, with negligible performance loss unless the peak power is overly restrictive.

\section{Conclusion}\label{Sec:Conclusion}
In this work, we reinterpreted SC from the perspective of the \textit{encoder--SF channel--decoder} pipeline. 
Unlike conventional approaches that assume a fixed SF channel, we observed that the SF channel is configurable and can be optimized to improve task performance under a mutual information constraint. 
We first provided analytical evidence for Gaussian sources and linear enc-dec mappings, which confirmed the validity of the joint optimization problem and revealed that the optimal SF channel allocates lower noise variance to sources with higher variance. 
After confirming this validity, we presented an end-to-end optimization strategy that jointly trains the DNN-based enc-dec and the differentiable surrogate model, applicable to both analog and digital SCs. 
We also proposed a PHY calibration strategy that enables the trained SF channel to be realized in practical wireless environments by adaptively controlling PHY parameters, including transmit power and modulation levels. 
Simulation results across various datasets demonstrated that the proposed SF channel optimization consistently achieves superior image reconstruction quality and adaptability under diverse channel conditions. 

Future research may extend the proposed framework in several directions. First, generative models could improve robustness to source distribution shifts by leveraging rich semantic priors. Their transformer-based attention mechanisms may also have a meaningful connection to the optimized SF-channel noise variance. Second, advanced PHY techniques such as beamforming and non-orthogonal multiple access could enhance scalability, with deep learning helping address the resulting non-convex joint design problems. Finally, security is another important consideration for the SF channel, motivating differential protection schemes that provide task-critical SFs with additional resilience against eavesdropping.  

\bibliographystyle{IEEEtran}
\bibliography{Reference}
\end{document}